\documentclass{article}

\def\noheaderplainsetup{

\topmargin=0pt \headheight=0pt \headsep=0pt  \oddsidemargin=0pt \evensidemargin=0pt  \textheight=9.1truein \textwidth=6.5truein}   

\noheaderplainsetup

\usepackage{amsfonts}

\begin{document}


\newcommand{\clthree}{\mbox{\bf CL12}}
\newcommand{\cltw}{\mbox{\bf CL12}}
\newcommand{\clfour}{\mbox{\bf CL4}}
\newcommand{\arfour}{\mbox{\bf CLA4}} 
\newcommand{\arfive}{\mbox{\bf CLA5}} 
\newcommand{\arsix}{\mbox{\bf CLA6}}
\newcommand{\arseven}{\mbox{\bf CLA7}}
\newcommand{\pa}{\mbox{\bf PA}}


\newcommand{\intimpl}{\mbox{\hspace{2pt}$\circ$\hspace{-0.14cm} \raisebox{-0.043cm}{\Large --}\hspace{2pt}}}
\newcommand{\zero}{\mbox{\small {\bf 0}}}
\newcommand{\izero}{\mbox{\scriptsize {\bf 0}}}
\newcommand{\one}{\mbox{\small {\bf 1}}}
\newcommand{\ione}{\mbox{\scriptsize {\bf 1}}}
\newcommand{\successor}{\mbox{\hspace{1pt}\boldmath $'$}}

\newcommand{\elz}[1]{\mbox{$\parallel\hspace{-3pt} #1 \hspace{-3pt}\parallel$}} 
\newcommand{\elzi}[1]{\mbox{\scriptsize $\parallel\hspace{-3pt} #1 \hspace{-3pt}\parallel$}}
\newcommand{\emptyrun}{\langle\rangle} 
\newcommand{\oo}{\bot}            
\newcommand{\pp}{\top}            
\newcommand{\xx}{\wp}               
\newcommand{\legal}[2]{\mbox{\bf Lr}^{#1}_{#2}} 
\newcommand{\win}[2]{\mbox{\bf Wn}^{#1}_{#2}} 
\newcommand{\seq}[1]{\langle #1 \rangle} 
\newcommand{\code}[1]{\ulcorner #1 \urcorner}          


\newcommand{\pst}{\mbox{\raisebox{-0.01cm}{\scriptsize $\wedge$}\hspace{-4pt}\raisebox{0.16cm}{\tiny $\mid$}\hspace{2pt}}}
\newcommand{\pcost}{\mbox{\raisebox{0.12cm}{\scriptsize $\vee$}\hspace{-4pt}\raisebox{0.02cm}{\tiny $\mid$}\hspace{2pt}}}

\newcommand{\gneg}{\mbox{\small $\neg$}}                  
\newcommand{\mli}{\hspace{2pt}\mbox{\small $\rightarrow$}\hspace{2pt}}                      
\newcommand{\cla}{\mbox{$\forall$}}      
\newcommand{\cle}{\mbox{$\exists$}}        
\newcommand{\mld}{\hspace{2pt}\mbox{\small $\vee$}\hspace{2pt}}     
\newcommand{\mlc}{\hspace{2pt}\mbox{\small $\wedge$}\hspace{2pt}}   
\newcommand{\mlci}{\hspace{2pt}\mbox{\footnotesize $\wedge$}\hspace{2pt}}   
\newcommand{\ade}{\mbox{\large $\sqcup$}}      
\newcommand{\ada}{\mbox{\large $\sqcap$}}      
\newcommand{\add}{\hspace{2pt}\mbox{\small $\sqcup$}\hspace{2pt}}                     
\newcommand{\adc}{\hspace{2pt}\mbox{\small $\sqcap$}\hspace{2pt}} 
\newcommand{\adci}{\hspace{2pt}\mbox{\footnotesize $\sqcap$}\hspace{2pt}}              
\newcommand{\clai}{\forall}     
\newcommand{\clei}{\exists}        
\newcommand{\tlg}{\bot}               
\newcommand{\twg}{\top}               
\newcommand{\fintimpl}{\mbox{\hspace{2pt}$\bullet$\hspace{-0.14cm} \raisebox{-0.058cm}{\Large --}\hspace{-6pt}\raisebox{0.008cm}{\scriptsize $\wr$}\hspace{-1pt}\raisebox{0.008cm}{\scriptsize $\wr$}\hspace{4pt}}}
\newcommand{\col}[1]{\mbox{$#1$:}}


\newtheorem{theoremm}{Theorem}[section]
\newtheorem{factt}[theoremm]{Fact}
\newtheorem{definitionn}[theoremm]{Definition}
\newtheorem{lemmaa}[theoremm]{Lemma}
\newtheorem{propositionn}[theoremm]{Proposition}
\newtheorem{conventionn}[theoremm]{Convention}
\newtheorem{examplee}[theoremm]{Example}
\newtheorem{exercisee}[theoremm]{Exercise}
\newtheorem{thesiss}[theoremm]{Thesis}
\newenvironment{definition}{\begin{definitionn} \em}{ \end{definitionn}}
\newenvironment{theorem}{\begin{theoremm}}{\end{theoremm}}
\newenvironment{lemma}{\begin{lemmaa}}{\end{lemmaa}}
\newenvironment{fact}{\begin{factt}}{\end{factt}}
\newenvironment{proposition}{\begin{propositionn} }{\end{propositionn}}
\newenvironment{convention}{\begin{conventionn} \em}{\end{conventionn}}
\newenvironment{example}{\begin{examplee} \em}{\end{examplee}}
\newenvironment{thesis}{\begin{thesiss} \em}{\end{thesiss}}
\newenvironment{exercise}{\begin{exercisee} \em}{\end{exercisee}}
\newenvironment{proof}{ {\bf Proof.} }{\  \rule{2.5mm}{2.5mm} \vspace{.2in} }
\newenvironment{idea}{ {\bf Proof idea.} }{\  \rule{1.5mm}{1.5mm} \vspace{.15in} }
\newenvironment{subproof}{ {\em Proof.} }{\  \rule{2mm}{2mm} \vspace{.1in} }

\title{Introduction to clarithmetic II}
\author{Giorgi Japaridze}

\date{}
\maketitle

\begin{abstract} The earlier paper ``Introduction to clarithmetic I'' constructed an axiomatic system of arithmetic based on computability logic, and proved its soundness and extensional completeness with respect to polynomial time computability. The present paper elaborates three additional sound and complete systems in the same style and sense: one for polynomial space computability, one for elementary recursive time (and/or space) computability, and one for primitive recursive time (and/or space) computability.  
\end{abstract}

\noindent {\em MSC}: primary: 03F50; secondary: 03F30; 03D75; 03D15; 68Q10; 68T27; 68T30

\

\noindent {\em Keywords}: Computability logic; Interactive computation; Implicit computational complexity;  Game semantics; Peano arithmetic; Bounded arithmetic; Constructive logics  


\section{Introduction}\label{intr}

Being a continuation of \cite{cla4}, this article fully and heavily relies on the terminology, notation, conventions and technical results of its predecessor, with which the reader is assumed to be well familiar (the good news, however, is that, \cite{cla4}, in turn, is self-contained).  

Remember, from \cite{cla4}, the system $\arfour$ of arithmetic, both semantically and syntactically based on computability logic (CoL). Its language was that of Peano arithmetic ($\pa$) augmented with the choice conjunction $\adc$, choice disjunction $\add$, choice universal quantifier $\ada$ and choice existential quantifier $\ade$. On top of the standard Peano axioms, $\arfour$  had two extra-Peano axioms: $\ada x\ade y(y= x+ 1)$ and $\ada x\ade y(y= 2x)$, one saying that the  function $x+ 1$ is computable, and the other saying the same about the function $2x$. The only logical rule of $\arfour$ was Logical Consequence (LC), meaning that the logical basis for the system was 
the (sound and complete) fragment {\bf CL12} of CoL. And the only nonlogical rule of inference was the induction rule
\[\frac{F(0)\hspace{30pt} F(x)\mli F(2x)\hspace{30pt} F(x)\mli F(2x+ 1)}{F(x)},\]
with $F(x)$ --- that is, its choice quantifiers $\ada,\ade$ --- required to be polynomially bounded. The system was proven in \cite{cla4} to be sound and extensionally (representationally) complete with respect to polynomial time computability. 

The present paper constructs three new {\bf CL12}-based systems: $\arfive$, $\arsix$, $\arseven$ and proves their soundness and extensional completeness with respect to polynomial space computability, elementary recursive time (and/or space) computability, and primitive recursive time (and/or space) computability, respectively. While $\arfour$ was already  simple  enough, the above three systems are even more so. All of them only need $\ada x\ade y(y= x+ 1)$ as a single extra-Peano axiom. As before, the only logical rule is LC. And the induction rule (the only nonlogical rule) of each of these systems is 
\[\frac{F(0)\hspace{30pt} F(x)\mli F(x+ 1)}{F(x)}.\]
 The three systems differ from each other only in what (if any) conditions are imposed on the formula $F(x)$ of induction. In $\arfive$, as in $\arfour$, $F(x)$ is required to be  polynomially bounded. $\arsix$ relaxes this requirement and allows $F(x)$ to be an exponentially bounded formula. $\arseven$ takes this trend towards relaxation  to an extreme and imposes no restrictions on $F(x)$ whatsoever. This way, unlike $\arfour$, $\arfive$ and $\arsix$, theory $\arseven$ is no longer in the realm of bounded arithmetics.  

The simplicity and elegance of these systems is additional evidence for the naturalness and productiveness of the idea of  basing  complexity-oriented systems and bounded arithmetic in particular on CoL instead of classical logic, even if one is only concerned with functions rather than the more general class of all interactive computational problems. In \cite{Buss}, achieving representational completeness with respect to polynomial space computable functions required considering a second-order extension of classical-logic-based bounded arithmetic (similarly in \cite{bbb3} for certain other complexity classes). In our  case, on the other hand, a transition from polynomial time ($\arfour$) to polynomial space ($\arfive$) in remarkably smooth with no need for any changes in the underlying language or logic, and with only minimal syntactic changes in the nonlogical part (induction rule) of the system.  Among the virtues of CoL is that, as a logic, it remains the same regardless of for what purposes (polynomial time computability, polynomial space computability,  computability-in-principle, \ldots) it is used. CoL does not have {\em variations}, but rather  has various (conservative) {\em fragments},\footnote{Including what has been termed ``intuitionistic computability logic'' (studied in \cite{Japjsl,int1,Propint}), contrary to what this name may suggest. Unlike, say, intuitionistic linear logic, which is indeed a variation of (classical) linear logic, intuitionistic computability logic is merely a conservative fragment of CoL, obtained by restricting its logical vocabulary  to the choice operators and the ultimate reduction operator.} depending on what part of its otherwise very expressive language is considered.  The fragment dealt with in the present paper, as well as in its predecessor \cite{cla4}, as well as in its even earlier predecessors \cite{Japtowards,Japlbcs}, is the same: logic $\cltw$.

\subsection{Technical notes}\label{a18}
All  terminology and notation not redefined in this paper has the same meaning as in \cite{cla4}. And all of our old conventions from \cite{cla4} extend to the present context as well.

Additionally we agree that  a ``{\bf sentence}'' always means a sentence (closed formula) of the language of $\arfour$. Similarly for ``{\bf formula}'', unless otherwise specified or suggested by the context. 

The definition of a polynomially bounded formula given in Section 11 of \cite{cla4}  contained a minor technical error. The correct formulation, on which we shall subsequently rely and which was really meant throughout \cite{cla4}, is as follows. We say that a formula $F$ is {\bf polynomially bounded} iff every subformula $\ada x G(x)$ (resp. $\ade x G(x)$) of $F$ has the form   $\ada x ( S(x)\mli H(x))$ (resp. $\ade x ( S(x)\mlc H(x))$), where  $S(x)$ is a polynomial sizebound for $x$  none of whose free variables is bound by $\cla$ or $\cle$ within $F$.

In the context of a given play (computation branch) of an HPM $\cal M$, by the {\bf spacecost} of a given clock cycle $c$ we shall mean the number of cells ever visited by the work-tape head of $\cal M$ by time $c$. We extend the usage of this term from clock cycles to the corresponding configurations as well. 

As in the preceding paragraph, we will be using the informal term ``{\bf play}'' mostly in reference to a computation branch of a given machine,  but occasionally it should rather be understood as the run spelled by such a branch. The meaning will usually be clear from the context.  

\section{$\arfive$, a theory of polynomial space computability}\label{ss11}

The {\bf language} of theory $\arfive$ is the same as that of $\arfour$ --- that is, it is an extension of the language of $\pa$ through the additional binary connectives $\adc,\add$ and quantifiers $\ada,\ade$.  And the axiomatization of $\arfive$ is obtained from that of $\arfour$ by deleting Axiom 9 (which is now redundant) and replacing the $\arfour$-Induction rule by the following rule, which we call {\bf $\arfive$-Induction}: 
\[\frac{\ada \bigl(F(0)\bigr)\hspace{30pt} \ada\bigl( F(x)\mli F(x\successor)\bigr)}{\ada \bigl(F(x)\bigr)},\]
where $F(x)$ is any  polynomially bounded formula. 
Here we shall say that $\ada \bigl(F(0)\bigr)$ is  the {\bf basis} of induction, and $\ada \bigl(F(x)\mli F(x\successor)\bigr)$ is the {\bf  inductive step}.
 
To summarize, the nonlogical axioms  of $\arfive$ are those of $\pa$ (Axioms 1-7) plus one single additional axiom $\ada x\ade y(y= x\successor)$ (Axiom 8). There are no logical axioms. The only logical inference rule  is Logical Consequence (LC) as defined in Section   10 of \cite{cla4}, and the only nonlogical inference rule   is $\arfive$-Induction.

 The following fact establishes that the old Axiom 9 of $\arfour$ would indeed be redundant in $\arfive$: 
\begin{fact}\label{onesuc}
$\arfive\vdash \ada x\ade y (y= x\zero)$. 
\end{fact}

\begin{proof} Argue in $\arfive$. First, by $\arfive$-Induction on $x$, we want to show
\begin{equation}\label{m24a}   \ade z (|z|\leq |x|+ |y|\mlc z= x+ y) .
\end{equation}
The basis  $ \ade z (|z|\leq |0|+ |y|\mlc z= 0+ y)$ is obviously solved by choosing the value of $y$ for the variable $z$. To solve the inductive step 
\[  \ade z (|z|\leq |x|+ |y|\mlc z= x+ y)\mli \ade z (|z|\leq |x\successor|+ |y|\mlc z= x\successor+ y),\]
we wait till Environment selects a value $a$ for $z$ in the antecedent. Then, using Axiom 8, we calculate the value $b$ of $a\successor$, and choose $b$ for $z$ in the consequent. The resulting position shown below is true by $\pa$, so we win:
\[ |a|\leq |x|+ |y|\mlc a= x+ y\mli  |a\successor|\leq |x\successor|+ |y|\mlc a\successor= x\successor+ y.\]

By LC, from (the $\ada$-closure of) (\ref{m24a}) we immediately get $ \ada x\ada y\ade z (z= x+ y)$; the latter, in turn, again by LC, implies $\ada x \ade z (z= x+ x)$, whence, together with the $\pa$-provable $\cla x(x+ x= x\zero)$, by LC, we get the target $\ada x \ade y (y= x\zero)$.\vspace{-7pt} 
\end{proof}

\begin{fact}\label{m24f}
$\arfive\vdash \ada x\ade y(x= y\zero\add x= y\one)$. 
\end{fact}

\begin{proof} Argue in $\arfive$. By $\arfive$-Induction on $x$, we want to show $\ade y\bigl(|y|\leq |x|\mlc (x= y\zero\add x= y\one)\bigr)$, which immediately implies the target  $\ada x\ade y(x= y\zero\add x= y\one)$ by LC. 

The basis $\ade y\bigl(|y|\leq |0|\mlc (0= y\zero\add 0= y\one)\bigr)$ is solved by choosing $0$ for $y$ and then choosing the left $\add$-disjunct. 

To solve the inductive step $\ade y\bigl(|y|\leq |x|\mlc (x= y\zero\add x= y\one)\bigr)\mli \ade y\bigl(|y|\leq |x\successor|\mlc (x\successor = y\zero\add x\successor = y\one)\bigr)$, we wait till Environment chooses a constant $a$ for $y$ in the antecedent, and also chooses one of the two $\add$-disjuncts there. 

Suppose the left $\add$-disjunct is chosen in the antecedent. So, by now, the game has been brought down to 
\( |a|\leq |x|\mlc x= a\zero\mli \ade y\bigl(|y|\leq |x\successor|\mlc (x\successor = y\zero\add x\successor = y\one)\bigr).\) Then we choose the same $a$ for $y$ in the consequent, and further choose the right disjunct there. The resulting position
 \( |a|\leq |x|\mlc x= a\zero\mli |a|\leq |x\successor|\mlc  x\successor = a\one\)   is true (by $\pa$), so we win.

Now suppose the right $\add$-disjunct is chosen in the antecedent. So, by now, the game has been brought down to $  |a|\leq |x|\mlc  x= a\one \mli \ade y\bigl(|y|\leq |x\successor|\mlc (x\successor = y\zero\add x\successor = y\one)\bigr)$. Then we, using Axiom 8, compute the value of $a\successor$, choose that value for $y$ in the consequent, and further choose the left $\add$-disjunct there. The game will be  brought down to the true $  |a|\leq |x|\mlc  x= a\one \mli  |a\successor|\leq |x\successor|\mlc x\successor = (a\successor) \zero$, so, again, we win.
\end{proof}

In the sequel we will heavily yet usually only implicitly rely on the following fact, which allows us to automatically transfer to $\arfive$ all $\arfour$-provability results established in \cite{cla4}.

\begin{fact}\label{m24b}
Every $\arfour$-provable sentence is also $\arfive$-provable.
\end{fact}

\begin{proof} From Fact \ref{onesuc} we know that $\arfive$ proves the only axiom (Axiom 9) of $\arfour$ not present in $\arfive$. So, $\arfive$ proves all axioms of $\arfour$. And the rule of LC is the same in the two theories. Therefore, it only remains to show that $\arfive$ is closed under the  rule of $\arfour$-Induction. So, assume $F(x)$ is a polynomially bounded formula, and $\arfive$ proves (the $\ada$-closures of)  each of the following three premises of $\arfour$-Induction:
\begin{eqnarray}
 & F(0);   & \label{m24c}\\
& F(x)\mli F(x\zero);  & \label{m24d}\\
& F(x)\mli F(x\one). &  \label{m24e}
\end{eqnarray}
Our goal is to show that $\arfive$ proves (the $\ada$-closure of) $F(x)$, the conclusion of $\arfour$-Induction.

Argue in $\arfive$. By $\arfive$-Induction on $x$, we want to prove 
\begin{equation}\label{m24g}
\ada y\bigl(|y|\leq |x|\mli F(y)\bigr) .  
\end{equation}

The basis $\ada y\bigl(|y|\leq |0|\mli  F(y)\bigr)$ is obviously taken care of by (\ref{m24c}), according to which, after choosing $0$ for $y$, we know how to solve $F(0)$.  To solve the inductive step 
\begin{equation}\label{m24h}
\ada y\bigl(|y|\leq |x|\mli F(y)\bigr)\mli \ada y\bigl(|y|\leq |x\successor|\mli F(y)\bigr),
\end{equation}
we wait till Environment chooses a constant $a$ for $y$ in the consequent. Then, using Fact \ref{m24f}, we find the binary predecessor $b$ of $a$, and also figure out whether $a= b\zero$ or ($\add$) $a= b\one$.  In either case, we specify $y$ as $b$ in the antecedent. 

If $a= b\zero$,  by now (\ref{m24h}) is brought down to 
\begin{equation}\label{m24i}
\bigl( |b|\leq |x|\mli F(b)\bigr) \mli  \bigl(|b\zero|\leq |x\successor|\mli  F(b\zero)\bigr).
\end{equation}
From (\ref{m24d}), we also know how to win $ F(b)\mli F(b\zero)$. By applying copycat between the two $F(b)$s and two $F(b\zero)$s, we win (\ref{m24i}). 

The case of $a= b\one$ is similar, only relying on (\ref{m24e}) instead of (\ref{m24d}). Thus, (\ref{m24g}) is proven.

Now, the target $F(x)$ can be easily seen to be a logical consequence of (\ref{m24g}) and the $\pa$-provable  $\cla x(|x|\leq |x|)$. 
\end{proof}

\begin{theorem}\label{tt1}
An arithmetical problem has a polynomial space solution iff it is provable in $\arfive$. 

Furthermore, there is an efficient procedure that takes an arbitrary extended $\arfive$-proof of an arbitrary sentence $X$ and constructs a   
 solution of $X$ (of $X^\dagger$, that is) together with an explicit polynomial bound for its space complexity. 
\end{theorem}

\begin{proof} The soundness (``if'') part of this theorem will be proven in Section \ref{sectsound}, and the completeness (``only if'') part in Section 
\ref{sectcompl}.\vspace{-7pt}
\end{proof}

\section{The soundness of $\arfive$}\label{sectsound}
This section is devoted to proving the soundness part of Theorem \ref{tt1}. We will only focus on showing that  any $\arfive$-provable sentence   has a polynomial space solution. The ``furthermore'' clause of the theorem  also claims that such a solution, together with an explicit polynomial bound for its space complexity, can be constructed efficiently. We will not explicitly verify this claim because it can be immediately seen to be true   for   the same reasons as those pointed out at the end of Section 13 of \cite{cla4} when justifying the  similar claim for $\arfour$. 

Consider an arbitrary $\arfive$-provable sentence $X$. In showing that $X$ has a polynomial space solution, we proceed by induction on the length of its proof.

Assume $X$ is an axiom of $\arfive$. If $X$ is one of Peano axioms, then it is a true  elementary sentence and therefore is won by a machine that makes no moves and consumes no space.  
And if $X$ is  $\ada x\ade y(y= x\successor)$ (Axiom 8), then it is won by a machine that (for the constant $x$ chosen by Environment for the variable $x$)    computes  the value $a$ of $x+ 1$, makes the move $a$  and retires in a moveless infinite loop that consumes no space.

Next, suppose $X$ is obtained from premises $X_1,\ldots,X_n$ by LC. By the induction hypothesis, for each $i\in\{1,\ldots,n\}$, we already have a solution (HPM) ${\cal N}_i$ of $X_i$ together with an explicit polynomial bound  $\xi_i$ for the space complexity of ${\cal N}_i$. Of course, we can think of each such HPM ${\cal N}_i$ as an $n$-ary GHPM that ignores its inputs. Then,   by clause 2 of Theorem 10.1 of \cite{cla4}, we can (efficiently) construct a solution ${\cal M}(\code{{\cal N}_1},\ldots,\code{{\cal N}_n})$ of $X$, together with an explicit polynomial bound $\tau(\xi_1,\ldots,\xi_n)$ for its space complexity.  

Finally, for the long rest of this section, assume $X$ is (the $\ada$-closure of) $F(x)$, where $F(x)$ is a polynomially bounded formula, and $X$ is obtained by $\arfive$-Induction on $x$. So, the premises are (the $\ada$-closures of) $F(0)$ and $F(x)\mli F(x\successor)$. By the induction hypothesis, there are HPMs ${\cal N}$ and ${\cal K}$ --- with certain explicit polynomial bounds $\xi$ and $\zeta$ for their space complexities, respectively --- that solve these two premises, respectively.  Fix them. We want to construct a solution $\cal M$ of $F(x)$.  

Remember the notion of an ``unreasonable move'' from Section 13 of \cite{cla4}. ``{\bf Reasonable}'', as expected, will mean ``not unreasonable''. We will say that a player $\xx\in\{\pp,\oo\}$ has played a run or ``play'' $\Gamma$ {\bf reasonably} with respect to a formula/game $G$ iff $\Gamma$ does not contain any unreasonable (with respect to $G$, in the context of  $G$) moves by $\xx$. And we say that $\Gamma$ is a {\bf reasonable} run or play of $G$ iff both players have played $\Gamma$ reasonably with respect to $G$. We say that a given machine $\cal H$ plays a given formula/game $G$ {\bf reasonably} iff, whenever $\Gamma$ is a run generated by $\cal H$, $\cal H$ (in the role of $\pp$) has played $\Gamma$ reasonably with respect to $G$. 

As done in Section 13 of \cite{cla4}, we replace  ${\cal N},{\cal K}$   by their ``{reasonable counterparts}'' --- i.e., machines that play the corresponding games $F(0)$ and $F(x)\mli F(x\successor)$ reasonably --- ${\cal N}',{\cal K}'$ and  corresponding  explicit polynomial bounds $\xi',\zeta'$ for their space complexities. For simplicity, we further replace the two bounds $\xi',\zeta'$ by the  common  bound $\phi= \xi'+ \zeta'$ for the space complexities of both machines ${\cal N}'$  and ${\cal K}'$.

To describe  $\cal M$, assume $x,\vec{v}$ are exactly the free variables of $F(x)$ (the case of $F(x)$ having no free occurrences of $x$ is trivial and we exclude it from our considerations), so that, in an expanded form,  $F(x)$ can be rewritten as $F(x,\vec{v})$. 
At the beginning,  our $\cal M$  waits for Environment to choose constants for the free variables of $F(x,\vec{v})$.   
For the rest of this section, assume 
 $k$ is the  constant chosen for the variable $x$,
 and $\vec{c}$ are the constants chosen for $\vec{v}$. Since the case of $k= 0$ is straightforward and not worth considering separately, we will additionally assume that $k\geq 1$. From now on, we shall write $F'(x)$ as an abbreviation of $F(x,\vec{c})$.

Further, we shall write ${\cal H}_{0}$ as an abbreviation of the phrase ``${\cal N}'$  in the scenario where the adversary, at the beginning of the play, has chosen the constants $\vec{c}$ for the variables $\vec{v}$\hspace{2pt}''. So, for instance, when saying that ${\cal H}_0$ moves on cycle $t$, it is to be understood as that, in the above scenario, ${\cal N}'$ moves on cycle $t$. As we see, strictly speaking, ${\cal H}_0$ is not a separate ``machine'' but rather it is just ${\cal N}'$ in a certain partially fixed scenario. Yet, for convenience and with some rather innocent abuse of language, in the sequel we may terminologically and even conceptually treat ${\cal H}_0$ as if it was a machine in its own right --- namely, the machine that works just like ${\cal N}'$ does in the scenario where the adversary, at the beginning of the play, has chosen the constants $\vec{c}$ for the variables $\vec{v}$.  Similarly, for any $n\geq 1$, we will write 
${\cal H}_{n}$ for the ``machine'' that works just like ${\cal K}'$ does in the scenario where the adversary, at the beginning of the play, has chosen the constants $\vec{c}$ for the variables $\vec{v}$ and the constant $n- 1$ for the variable $x$. So, ${\cal H}_{0}$ (thought of as a machine) wins the constant game $F'(0)$ and, for each $n\geq 1$, ${\cal H}_{n}$ wins the constant game $F'(n-1)\mli F'(n)$.

In the same style as the notation ${\cal H}_n$ is used, we write ${\cal M}_k$ for the ``machine'' that works just like ${\cal M}$ does after the above event of Environment's having chosen $k$ and $\vec{c}$ for $x$ and $\vec{v}$, respectively. So, in order to complete our description of $\cal M$, it will suffice to simply define ${\cal M}_k$ and say that, after Environment has chosen constants for all free variables of $F(x)$, $\cal M$ continues playing as (``turns into'') ${\cal M}_k$. Correspondingly,  in showing that $\cal M$ wins $\ada F(x)$, it will be sufficient to show that ${\cal M}_k$ wins $F'(k)$. 

The idea underlying the work of ${\cal M}_k$ can be summarized by saying that what ${\cal M}_k$ does is synchronization --- in the sense explained in Section 13 of \cite{cla4} --- between $k+ 2$ games, real or imaginary (simulated). Namely:
\begin{itemize}
\item It synchronizes the imaginary play of $F'(0)$ by ${\cal H}_{0}$ with the antecedent of the imaginary play of $F'(0)\mli F'(1)$ by ${\cal H}_{1}$.
\item For each $n$ with $1\leq n< k$, it synchronizes the consequent of the imaginary play of $F'(n- 1)\mli F'(n)$ by ${\cal H}_{n}$ with the 
antecedent of the imaginary play of $F'(n)\mli F'(n+ 1)$ by ${\cal H}_{n+ 1}$.
\item It (essentially) synchronizes the consequent of the imaginary play of  $F'(k- 1)\mli F'(k)$ by ${\cal H}_{k}$ with the real play of $F'(k)$.
\end{itemize}
Therefore, since ${\cal H}_{0}$ wins $F'(0)$ and each ${\cal H}_{n}$ with $1\leq n\leq k$ wins $F'(n- 1)\mli F'(n)$, ${\cal M}_k$ wins $F'(k)$ and thus $\cal M$ wins $F(x)$, as desired. 

In section 13 of \cite{cla4},  synchronization in the above style was achieved by simulating all imaginary plays in parallel. The present case does not allow us doing the same though, and synchronization should be conducted in a more careful way. Namely, a parallel simulation of all plays is no longer possible, because there are exponentially (in the size of the binary representation of $k$) many simulations to perform, which would require an exponential amount of space. So, instead, simulations in the present case should be performed is some sequential rather than parallel manner, with subsequent simulations recycling the space used by the previous ones, and with the overall procedure keeping forgetting the results of most previous simulations and re-computing the same information over and over many times. This is a typical case of trading time for space. We postpone our description of how ${\cal M}_k$ exactly works till Subsection \ref{sn}, after having elaborated all necessary preliminaries in Subsections \ref{sagree}-\ref{saggr}.     

\subsection{$\mathfrak{L}$ and some other important quantities}\label{sagree}
We agree that, throughout the rest of Section \ref{sectsound}:
\begin{itemize}
\item $\mathfrak{l}$ denotes the size of the greatest constant among $k,\vec{c}$. 
\item $\mathfrak{e}$ denotes the maximum number of $\oo$-labeled moves in any legal run of $F'(0)$. 
\item $\mathfrak{d}$ denotes  the maximum number of whatever-labeled moves in any legal run of $\ada F(x)$. 
\item $\mathfrak{q}$  denotes the total number of symbols that may ever appear on either tape  of the machines ${\cal N}'$ and ${\cal K}'$. 
\item $\mathfrak{s}$  denotes the total number of states  of the machines ${\cal N}'$ and ${\cal K}'$.
\end{itemize}

It follows from our assumptions regarding ${\cal N}'$ and ${\cal K}'$ that all ${\cal H}_n$ ($0\leq n\leq k$) play the corresponding games $F'(0),\ F'(0)\mli F'(1),\ \ldots, F'(k-1)\mli F'(k)$  
 legally\footnote{As easily understood, by saying that a given player plays a given game $G$ {\bf legally} we mean that it does not make any illegal moves of $G$ in legal positions of $G$.} and reasonably. If we additionally assume that so does their adversary, then, as was done in \cite{cla4}, one can easily write a term $\eta(w)$ with a single variable $w$ such that the sizes of   moves ever made by either player in any of the above 
 games  never exceed $\eta(\mathfrak{l})$. For instance, if $F(x)$ is \[\ade u\bigl(|u|\leq |x|\times |z|\mlc\ada v(|v|\leq |u|+ |x|\mli G)
\bigl)\] where $G$ is elementary, then $\eta(w)$ can be taken to be $w\times w+ w+ 0\successor\successor\successor\successor$.  Fix this $\eta$. 

In the following three lemmas, in the context of a given ${\cal H}_n\in\{{\cal H}_0,\ldots,{\cal H}_k\}$, ``the corresponding game'' should be understood as $F'(0)$ if $n=0$, and as $F'(n-1)\mli F'(n)$ if $n\geq 1$. 

Since both ${\cal N}'$ and ${\cal K}'$ run in space $\phi$, we obviously have:
\begin{lemma}\label{m29aa}
For any $n\in\{0,\ldots,k\}$, in any play by ${\cal H}_n$, as long as ${\cal H}_n$'s  adversary
plays legally and reasonably with respect to the corresponding game, the spacecost of no clock cycle exceeds $\phi\bigl(\eta(\mathfrak{l})\bigr)$. 
\end{lemma}

By  the {\bf symbolwise length} of a position $\Phi$ we shall mean  the number of cells that $\Phi$ takes when spelled on the run tape. Taking into account that the sizes of moves by either player in legal and reasonable plays of   $F'(0)$, $F'(0)\mli F'(1)$, \ldots, $F'(k- 1)\mli F'(k)$  never exceed $\eta(\mathfrak{l})$ and that at most $2\mathfrak{d}$ moves can be (legally) made in those plays, we obviously   have: 
\begin{lemma}\label{m29aaa}
For any $n\in\{0,\ldots,k\}$, at any time in any play by ${\cal H}_n$, as long as ${\cal H}_n$'s  adversary
plays  legally and reasonably with respect to the corresponding game, the symbolwise length of the position
spelled on the run tape of ${\cal H}_n$ does not exceed $2\mathfrak{d}\eta(\mathfrak{l})+ 2\mathfrak{d}$.\footnote{Here ``$+ 2\mathfrak{d}$'' is to account for the labels $\twg$ and $\tlg$ attached to moves.}
\end{lemma} 

We let $\mathfrak{L}$ be an  abbreviation defined by 
\[\mathfrak{L} \ =\ \mathfrak{s}\times \Bigl(\phi\bigl(\eta(\mathfrak{l})\bigr)\Bigr)\times\Bigl(2\mathfrak{d} \eta(\mathfrak{l})+ 2\mathfrak{d}+ 1\Bigr)\times \Bigl(\mathfrak{q}^{\phi(\eta(\mathfrak{l}))} \Bigr)\times \Bigl(\mathfrak{q}^{2\mathfrak{d} \eta(\mathfrak{l})+ 2\mathfrak{d}+ 1} \Bigr).\]

\begin{lemma}\label{m29a}
Consider any machine ${\cal H}_n\in\{{\cal H}_0,\ldots,{\cal H}_k\}$, and any cycle (step, time) $c$ of any play by ${\cal H}_n$.
If the adversary of ${\cal H}_n$ plays legally and reasonably with respect to the corresponding game, and it does not move  at any time $d$ with 
$d\geq c$, then ${\cal H}_n$ does not move at any
 time $d$ with $d\geq c+\mathfrak{L}$.
\end{lemma}

\begin{proof} Consider any  play by any machine ${\cal H}_n\in\{{\cal H}_0,\ldots,{\cal H}_k\}$, where both players have played legally and reasonably, and answer the following question: How many different configurations of ${\cal H}_n$ are there that may emerge in the play? There are at most $\mathfrak{s}$ possibilities for the state of such a configuration. These possibilities are accounted for by the 1st of the five factors of $\mathfrak{L}$. Next, in view of Lemma \ref{m29aa}, there are at most $\phi\bigl(\eta(\mathfrak{l})\bigr)$ possible locations of the work-tape head.  This number is accounted for by the 2nd factor of $\mathfrak{L}$. Next, in view of Lemma \ref{m29aaa}, there are at most $2\mathfrak{d} \eta(\mathfrak{l})+ 2\mathfrak{d}+ 1$ possible locations of the run-tape head,\footnote{Remember that a scanning head of an HPM can never move beyond the leftmost blank cell. So, ``$2\mathfrak{d} \eta(\mathfrak{l})+ 2\mathfrak{d}$'' is to account for the position spelled on the tape, and ``$+1$'' is to account for the possibility of visiting the blank cell following that position.} and this number is accounted for by the 3rd factor of $\mathfrak{L}$. Next, in view of Lemma \ref{m29aa}, obviously there are at most $\mathfrak{q}^{\phi(\eta(\mathfrak{l}))}$ possible contents of the work tape, and this number is accounted for by the 4th factor of $\mathfrak{L}$. Finally, in view of Lemma \ref{m29aaa}, there are at most $\mathfrak{q}^{2\mathfrak{d} \eta(\mathfrak{l})+ 2\mathfrak{d}+ 1}$ possible contents of the run tape, and this number is accounted for by the 5th factor of $\mathfrak{L}$. Thus, there are at most $\mathfrak{L}$ possible configurations. Now, consider the scenario where the adversary of ${\cal H}_n$ makes no moves beginning from a clock cycle $c$. Assume, for a contradiction, that ${\cal H}_n$ makes a move $\alpha$ at some time $d$ with $d> c+ \mathfrak{L}$. Since there are fewer that $d- c$ configurations, some configuration should repeat itself between the steps $c$ and $d$. In other words, ${\cal H}_n$ is in an infinite loop. Hence, it will make the same move $\alpha$ over and over again, which means that ${\cal H}_n$ does not play legally (there are at most $\mathfrak{d}$ legal moves by $\pp$ in the play), contrary to our assumptions.  
\end{proof}

\subsection{The procedure $\mathbb{SIM}$}
By a 
{\bf movesequence} we shall mean a (possibly empty) finite sequence $\vec{\alpha}=\seq{\alpha_1,\ldots,\alpha_r}$ of 
moves, and a {\bf body} means a (possibly empty) tuple $B=(\vec{\alpha}_1,\ldots,\vec{\alpha}_s)$ of movesequences. 
 The number $s$ is said to be the {\bf size} of such a body $B$. 
A {\bf signed movesequence} $S$ is $-\vec{\omega}$ (in which case we say that $S$ is {\bf negative}) or $+\vec{\omega}$ (in which case we say that $S$ is {\bf positive}), where $\vec{\omega}$ is a 
movesequence. In many contexts we may terminologically identify a signed movesequence $-\vec{\omega}$ or $+\vec{\omega}$ with its $\vec{\omega}$ part. For instance, we may say ``$+\vec{\omega}$ is nonempty'', which should be understood as that $\vec{\omega}$ is nonempty.

Our ${\cal M}_k$ simulates the work of a given machine ${\cal H}_n$ ($0\leq n\leq k$) through running the procedure $\mathbb{SIM}_n$ defined below. This procedure takes a pair $(B_1,B_2)$ of bodies as an argument, where $B_1$ is nonempty,
and returns a signed movesequence $S$. We indicate this relationship by writing $\mathbb{SIM}_n(B_1,B_2)=S$. When $n=0$, $B_1$  is required to be $(\seq{})$, which  makes this argument a ``dummy'' one; also, the output of  $\mathbb{SIM}_0$ is always positive, which makes the sign $+$ of that output also ``dummy''. 

We first take a brief and informal preliminary look at $\mathbb{SIM}_n$, starting with the simpler case of  $\mathbb{SIM}_0((\seq{}),B_2)$. 
Let $B_2=(\vec{\beta_1},\ldots,\vec{\beta}_b)$. This argument describes a behavior --- moves made by --- ${\cal H}_0$'s imaginary adversary and, this way, determines the scenario of the work of ${\cal H}_0$ that needs to be simulated. Namely, this is the scenario where the adversary made the moves of $\vec{\beta}_1$ (``{\bf moved $\vec{\beta}_1$}'' for short), all at once, on clock cycle $\mathfrak{L}$, then moved $\vec{\beta}_2$ on cycle $2\mathfrak{L}$, $\vec{\beta}_3$ on cycle $3\mathfrak{L}$, and so on. Then $\mathbb{SIM}_0((\seq{}),B_2)=+\vec{\psi}$, where $\vec{\psi}$ is the sequence of all moves made by ${\cal H}_0$ during the $\mathfrak{L}$ cycles following the adversary's last chunk $\vec{\beta}_b$ of moves --- that is, the sequence of moves made by ${\cal H}_0$ during the cycles 
$b\mathfrak{L}$ (including) through $(b+ 1)\mathfrak{L}$ (not including). We denote such an interval by $[b \mathfrak{L},\ldots,(b+ 1)\mathfrak{L})$, with ``$[$'' indicating that $b\mathfrak{L}$ is included, and ``$)$'' indicating that $(b+ 1)\mathfrak{L}$ is not included. Thus, the output of $\mathbb{SIM}_0((\seq{}),B_2)$ does not describe the full behavior of (all moves made by) ${\cal H}_0$ in the scenario determined by $B_2$, but rather only the moves made during the last one of the $\mathfrak{L}$-step-long ``episodes'' of that scenario. Why $\mathfrak{L}$-step-long, intuitively speaking? That is because, in view of Lemma \ref{m29a}, provided that the adversary plays legally and reasonably, $\mathfrak{L}$ steps are sufficient for ${\cal H}_0$ to make all moves that it was ``ever going to make'' in response to it's adversary's actions. That is, adding any extra amount of steps to the last episode will not result in a different value of $\vec{\psi}$.     

The case of $\mathbb{SIM}_n(B_1,B_2)$ with $1\leq n\leq k$ is similar but somewhat more complicated. In fact, $\mathbb{SIM}_0((\seq{}), B_2)$ is a special case of $\mathbb{SIM}_n(B_1,B_2)$ if we (assume that $n=0$, $B_1=(\seq{})$ and) think of $F'(0)$ as the implication $F'(-1)\mli F'(0)$ with the ``dummy'' antecedent $F'(-1)=\twg$. As in the preceding case, 
 the argument $(B_1,B_2)$ determines the scenario of the work of ${\cal H}_n$ that needs to be simulated. In this scenario, the moves of $B_1$ are ones made by ${\cal H}_n$'s adversary in the antecedent of $F'(n-1)\mli F'(n)$, and the moves of $B_2$ are ones made in the consequent. The moves of the first movesequence of $B_1$ are always assumed to be made at the very beginning of the play, i.e., on clock cycle $0$. The order in which the rest of the moves are ``imagined'' to be made depends on how things evolve, namely, on whether ${\cal H}_n$ responds by a nonempty or an empty sequence of moves in the antecedent of $F'(n-1)\mli F'(n)$. In the former (resp. latter) case, the next series of ${\cal H}_n$'s adversary's imaginary moves will be the first not-yet-fetched movesequence of $B_1$ (resp. $B_2$). As in the case of $n=0$, the simulation of ${\cal H}_n$ proceeds episode-by-episode, with each episode lasting $\mathfrak{L}$ steps. Namely,  the $i$th episode covers steps  $[(i- 1)\mathfrak{L},\ldots,i \mathfrak{L})$. The overall procedure ends when it tries to fetch the next not-yet-fetched movesequence of either $B_1$ or $B_2$ but finds that there are no such movesequences remaining. In the former case the output $S$ is stipulated to be $-\vec{\nu}$, where 
$\vec{\nu}$ is the sequence of moves made by ${\cal H}_n$ in the antecedent of $F'(n-1)\mli F'(n)$ during the last episode of simulation. And in the latter case $S$ is $+\vec{\psi}$, where $\vec{\psi}$ is the series of moves made by ${\cal H}_n$ in the consequent of $F'(n-1)\mli F'(n)$ since the last movesequence of $B_2$ was fetched. 

In precise terms, this is how the procedure (that computes the value/output of) $\mathbb{SIM}_0((\seq{}),B_2)$ works. It creates an 
integer-holding variable $y$ initialized to $0$, and two signed-movesequence-holding variables $S$ and $R$, with $S$ having no initial value and $R$ initialized to $+\seq{}$.\footnote{The presence of the variable $S$ may seem redundant at this point, as $\mathbb{SIM}_0((\seq{}),B_2)$ (and likewise  $\mathbb{SIM}_n(B_1,B_2)$ with $n\geq 1$) could be defined in a simpler way without it. The reason why we want to have $S$ will become clear in Subsection \ref{stf}. Similarly, we could have done without the variable $R$ as well --- it merely serves the purpose of ``synchronizing'' the cases of $n=0$ and $n\geq 1$. Similar reasons also explain our obviously unnecessary/dummy usage of the sign ``$+$'' in the present case.}   After this initialization step, the procedure goes into the following loop LOOP.
Each iteration of the latter simulates certain $\mathfrak{L}$ steps of $H_n$, and the subsequent iteration (if one exists) continues simulation from the point where the previous one stopped.  

\begin{quote}
{\em LOOP}:  Let $+\vec{\omega}$ be the value of $R$ ($R$ never takes negative values when $n=0$). Simulate\footnote{That is, continue the simulation performed during the preceding iterations of LOOP if such iterations exist. The same comment applies to the description of $\mathbb{SIM}_n$ for $n\geq 1$ given shortly.} steps $[y\mathfrak{L},\ldots,(y+ 1)\mathfrak{L})$ of ${\cal H}_0$ in the scenario where, at step $y\mathfrak{L}$, the adversary moved $\vec{\omega}$, and made no other moves. Let  $\vec{\psi}$ be the moves made by ${\cal H}_0$  during the above steps  $[y\mathfrak{L},\ldots,(y+ 1)\mathfrak{L})$. 
Set the value of $S$ to $+\vec{\psi}$. Then, if $y$ equals the size of $B_2$, return $S$. Otherwise, increment $y$ to $y+ 1$, set $R$ to the $y$th movesequence of $B_2$ prefixed with ``$+$'', and repeat LOOP. 
\end{quote}

Next, this is  how the procedure $\mathbb{SIM}_n(B_1,B_2)$ exactly works when  $n\geq 1$. It creates three integer-holding variables $y_1,y_2,z$, with $y_1$ initialized to $1$, $y_2$ to $0$ and $z$ to $0$.\footnote{Intuitively, $y_1$ keeps track of how many movesequences of $B_1$ have been fetched so far, $y_2$ does the same for $B_2$, and $z$ tells us how many $\mathfrak{L}$-step-long episodes had already been simulated by the time when the latest movesequence of $B_2$ was fetched.}  It further creates two signed-movesequence-holding variable $S$ and $R$, 
with  $S$ having no initial value and $R$  
initialized to $-\vec{\alpha}$, where $\vec{\alpha}$ is the first movesequence of $B_1$ (remember that $B_1$ is required to be nonempty). After this initialization step, the procedure goes into the following loop LOOP.
As before, each iteration of the latter simulates certain $\mathfrak{L}$ steps of ${\cal H}_n$ (namely, the $i=(x+y)$th iteration simulates $\mathfrak{L}$ steps starting from step $(i-1)\mathfrak{L}$), and the subsequent iteration (if one exists) continues simulation from the point where the previous one stopped.  

\begin{quote}
{\em LOOP}:  Let $-\vec{\omega}$ (resp. $+\vec{\omega}$) be the value of $R$.   Simulate steps $[(y_1+ y_2-1)\mathfrak{L},\ldots,(y_1+ y_2)\mathfrak{L})$ of ${\cal H}_n$ in the scenario where, at step $(y_1+ y_2-1)\mathfrak{L}$, the adversary moved $\vec{\omega}$ in the antecedent (resp. consequent) of $F'(n-1)\mli F'(n)$, and made no other moves. Let $\vec{\nu}$ be the moves made by ${\cal H}_n$ in the antecedent during the above steps  $[(y_1+ y_2-1)\mathfrak{L},\ldots,(y_1+ y_2)\mathfrak{L})$. And let $\vec{\psi}$ be the moves made by ${\cal H}_n$ in the consequent during steps  $[z\mathfrak{L},\ldots,(y_1+ y_2)\mathfrak{L})$. 
\begin{itemize}
  \item If $\vec{\nu}$ is nonempty, set the value of $S$ to $-\vec{\nu}$. Then, if $y_1$ equals the size of $B_1$, return $S$; otherwise, increment $y_1$ to $y_1+ 1$, set $R$ to the $y_1$th movesequence of $B_1$ prefixed with ``$-$'', and repeat LOOP. 
  \item If $\vec{\nu}$ is empty, set the value of $S$ to $+\vec{\psi}$. Then, if $y_2$ equals the size of $B_2$, return $S$. Otherwise, increment $y_2$ to $y_2+ 1$, set $R$ to the $y_2$th movesequence of $B_2$ prefixed with ``$+$'', update $z$ to 
 $y_1+ y_2-1$, and repeat LOOP. 
\end{itemize}
\end{quote}
 
We say that a body $(\vec{\alpha}_1,\ldots,\vec{\alpha}_{s})$ is an {\bf extension} of a body 
 $(\vec{\beta}_1,\ldots,\vec{\beta}_t)$ iff $t\leq s$ and $\vec{\alpha}_1=\vec{\beta}_1,\ldots,\vec{\alpha}_t=\vec{\beta}_t$.
 
\begin{lemma}\label{golemma}
Consider any $n$ with $1\leq n\leq k$ and any two bodies $B$ and $C$, where $B$ is nonempty.

1. If $\mathbb{SIM}_n(B,C)$ is positive, then, for every extension $B'$ of $B$,  $\mathbb{SIM}_n(B',C)=\mathbb{SIM}_n(B,C)$. 

2. If $\mathbb{SIM}_n(B,C)$ is negative, then, for every extension $C'$ of $C$, $\mathbb{SIM}_n(B,C')=\mathbb{SIM}_n(B,C)$.

3. Whenever  $\mathbb{SIM}_n(B,C)$ is negative, the size of $B$ does not exceed $\mathfrak{e}$. 
\end{lemma}

\begin{proof} Clauses 1-2 can be verified through a straightforward analysis of the work of $\mathbb{SIM}_n$. For clause 3, assume $\mathbb{SIM}_n(B,C)=-\vec{\omega}$, and let $s$ be the size of $B$. Observe that, in the process of computing $\mathbb{SIM}_n(B,C)$, all negative values that the variable $S$ ever takes, including its last value $-\vec{\omega}$, are nonempty. All such negative values consist of moves made by ${\cal H}_n$ in the antecedent of $F'(n-1)\mli F'(n)$, where ${\cal H}_n$ plays in the role of $\oo$. From the work of $\mathbb{SIM}_n$ we can see that altogether there are $s$ such values. 
Remember that $\mathfrak{e}$ is the maximum number of $\oo$-labeled moves in any legal run of $F'(n-1)$. So, we may (retroactively) assume that 
${\cal H}_n$ never makes an $(\mathfrak{e}+1)$th move in the antecedent of  $F'(n-1)\mli F'(n)$, 
because making such an always-illegal move is pointless.\footnote{Otherwise, if our assumption is not satisfied, ${\cal H}_n$  (${\cal K}'$, to be more accurate) 
can be easily modified so as to satisfy it while still winning the corresponding game. Such a modification 
would only impose a constant --- and hence safely ignorable --- space overhead on the work of the machine.} This means that  $s\leq \mathfrak{e}$, as desired.\end{proof}

\subsection{Aggregations}\label{saggr}
For a body $B=(\vec{\alpha}_1,\ldots,\vec{\alpha}_s)$, we will write $B^{\mbox{\em odd}}$ (resp. $B^{\mbox{\em even}}$) to denote the body $(\vec{\alpha}_1,\vec{\alpha}_3,\ldots)$ (resp. $(\vec{\alpha}_2,\vec{\alpha}_4,\ldots)$) obtained from $B$ by deleting each $\vec{\alpha}_i$ with an even (resp. odd) $i$.   

By an {\bf entry} we shall mean a pair $E=[n,B]$, where $n$, called the {\bf index} of $E$ (and, correspondingly,   $E$ said to be {\bf $n$-indexed}), is an element of $\{0,\ldots,k\}$, and $B$, called the {\bf body of $E$}, is a nonempty body. The {\bf size} of an entry $E$ should be understood as the size of its body.

An {\bf aggregation} is a (possibly empty) sequence  
$A=\seq{E_1,\ldots,E_r}$ of entries such that:
\begin{description}
  \item[(i)\ ] \ The indices of the entries of $A$ are strictly increasing. That is, the index of any given entry is strictly smaller than the index of any entries to the right of it.
\item[(ii)\ ] Each odd-size entry is to the left of each even-size entry. 
\item[(iii)] The sizes of the odd-size entries are strictly decreasing. That is, the size of any odd-size entry is strictly smaller than the size of any (odd-size) entry to the left of it.
\item[(iv)] The sizes of the even-size entries are strictly increasing. That is, the size of any even-size entry is strictly smaller than the size of any (even-size) entry to the right of it.
\end{description}  

We say that an aggregation is {\bf passive} iff it has a $k$-indexed odd-size entry.
Otherwise it is {\bf active}.

The {\bf activity triple} of an active aggregation $A$ is $(n,L,R)$, where: 
\begin{enumerate}
  \item $n$ is the smallest element of $\{0,\ldots,k\}$ which is greater than the index of any odd-size entry of $A$. 
  \item If $A$ has no entry whose index is $n-1$,\footnote{Which, by condition 1, is the same as to say that $n=0$.}  then $L$ is the body $(\seq{})$. Otherwise, $L$ is the body of the $(n-1)$-indexed (and thus the rightmost odd-size) entry of $A$. 
  \item If $A$ has no entry whose index is $n$, then $R$ is the empty body $()$. Otherwise, 
$R$ is the body of the $n$-indexed (and thus the leftmost even-size) entry of $A$.  
\end{enumerate}

\subsection{The loop $\mathbb{MAIN}$ and the work of ${\cal M}_k$}\label{sn}

Now we are ready to finalize our description of the work of ${\cal M}_k$. This is a machine that creates an aggregation-holding variable $A$, initializes it to the empty aggregation $\seq{}$, and then goes into the loop $\mathbb{MAIN}$ described below. 

\begin{quote} $\mathbb{MAIN}$:  Act depending on whether $A$ is active or passive.  
 
{\em Case 1}: $A$ is active.   Compute the value $S$ of $\mathbb{SIM}_n(L^{\mbox{\em odd}},R^{\mbox{\em even}})$, where $(n,L,R)$ is the activity triple of $A$. Then act depending on whether $S$ is positive or negative.

{\em Subcase 1.1}: $S$ is positive, namely, $S=+\vec{\omega}$. Then:
\begin{description} 
  \item[(i)] \ If $A$ has a $n$-indexed entry, then modify $A$ by adding $\vec{\omega}$ as a new (last) movesequence to the body of that entry. 
  \item[(ii)] Otherwise, modify $A$ by inserting into it the entry $[n,(\vec{\omega})]$ so that it is to the right of all (old) odd-size entries and to the left of all even-size entries.
\end{description}  
In either case, let $A'$ be the resulting list of entries, and $E$ be its $n$-indexed entry. Delete all odd-size entries in $A'$ other than $E$ whose sizes are not greater than the size of $E$.  Now update $A$ to the resulting aggregation, and repeat $\mathbb{MAIN}$.
  
{\em Subcase 1.2}: $S$ is negative, namely, $S=-\vec{\omega}$. Notice that then $n>0$ (because $\mathbb{SIM}_0$ never returns a negative value), and $A$ has an $(n-1)$-indexed odd-size entry $E$. Modify $A$ by adding $\vec{\omega}$ as a new (last) movesequence  to the body of $E$.   
 Then, in the resulting list of entries, delete all even-size entries other than (the updated) $E$ whose sizes are not greater than the size of (the updated) $E$.  Now update $A$ to the resulting aggregation, and repeat $\mathbb{MAIN}$.

{\em Case 2}: $A$ is passive.  Let $B=(\vec{\alpha}_1,\ldots,\vec{\alpha}_s)$ be the body of the last ($k$-indexed) entry of $A$. Scan the run tape and count the total number of $\pp$-labeled moves on it. If that number is smaller than the total number of moves in (all movesequences of) $B^{\mbox{\em odd}}$, make the moves of $\vec{\alpha}_s$ in the real play, in the same order as they appear in $\vec{\alpha}_s$; if not, make no moves. In either case, then poll the run tape again to see if it contains an  $((s+1)/2)$th $\oo$-labeled move $\theta$. If not, continue polling repeatedly, looking for such a $\theta$. If and when such a $\theta$ is found,  check if it is a legal move of $F'(k)$ in the corresponding position (i.e., in the position consisting of the labmoves listed on the run tape on the left of $\theta$). If $\theta$ is illegal, retire, i.e. go into an infinite loop that consumes no space and makes no moves. Suppose now $\theta$ is legal. Let $\omega$ be $\theta$ if the latter is reasonable with respect to $F'(k)$ in the corresponding position. Otherwise, if $\theta$ is unreasonable, it must have a suffix ``$.c$'' for a certain ``unreasonably long'' constant $c$; in this case, let $\omega$ be the result of replacing in $\theta$ the suffix $c$ by $0$. In either case, update $A$ by adding  $\seq{\omega}$ as a new (last) movesequence to the body of the $k$-indexed entry,  and repeat $\mathbb{MAIN}$. 
\end{quote}

\subsection{The adequacy of  ${\cal M}$}\label{stf}
 Our main purpose now is to verify that ${\cal M}_k$ indeed wins $F'(k)$ and hence ${\cal M}$ wins $F(x)$. The polynomial space complexity of $\cal M$ will be established at the very end of this subsection.

For the rest of the present subsection, when analyzing the work and behavior of ${\cal M}_k$, we will implicitly have some arbitrary but fixed computation branch (``play'') of ${\cal M}_k$ in mind. So, for instance, when we say ``the $i$th iteration of $\mathbb{MAIN}$'', it should be understood in the context of that branch.   

In what follows, $I$ will stand for the set of positive integers $i$ such that  $\mathbb{MAIN}$ is iterated at least $i$ times. Next, for each $i\in I$, $A_i$ will stand for the value of the aggregation/variable $A$ at the beginning  of the $i$th iteration of $\mathbb{MAIN}$. 

\begin{lemma}\label{beijing}
$I$ is finite, i.e., $\mathbb{MAIN}$ is iterated only a finite number of times.
\end{lemma}

\begin{proof} We first claim that 
\begin{equation}\label{bodycond}
\mbox{\em For any $i\in I$ and any entry $E$ of $A_i$, the size of $E$ does not exceed $2\mathfrak{e}+ 1$.} 
\end{equation}
To verify (\ref{bodycond}), deny it for a contradiction. Let then $i$ be the smallest number in $I$ such that $A_i$ has 
 an $(2\mathfrak{e}+2)$-size, $n$-indexed entry $[n,(\vec{\alpha}_1,\ldots,\vec{\alpha}_{2\mathfrak{e}+2})]$ --- it is not hard to see that such an $i$ exists, and $i> 1$ because $A_1$ has no entries. The only way the above entry could have emerged in $A_i$ is that $A_{i-1}$ contained the entry $[n,(\vec{\alpha}_1,\ldots,\vec{\alpha}_{2\mathfrak{e}+1})]$, and its body ``grew'' into $(\vec{\alpha}_1,\ldots,\vec{\alpha}_{2\mathfrak{e}+2})$ on the transition from $A_{i-1}$ to $A_i$.  
This, in turn, obviously means that the activity triple of $A_{i-1}$ was $(n+1,(\vec{\alpha}_1,\ldots,\vec{\alpha}_{2\mathfrak{e}+1}), B)$ for a certain body $B$, and $\mathbb{SIM}_{n+1}((\vec{\alpha}_1,\ldots,\vec{\alpha}_{2\mathfrak{e}+1})^{\mbox{\em odd}}, B^{\mbox{\em even}})=-\vec{\alpha}_{2\mathfrak{e}+2}$. This, however, is impossible by clause 3 of Lemma \ref{golemma}. 

Next, we define the binary relation $\prec$ on aggregations by stipulating that $B\prec A$ iff there is a positive integer $s$ such that the following three conditions are satisfied:
\begin{enumerate}
  \item $A$ has an entry $E_a$ of size $s$.
  \item  If $B$ has an entry $E_b$ of size $s$, then:
  \begin{enumerate}  
    \item if $s$ is odd, then the index of $E_a$ is greater than the index of $E_b$;
    \item if $s$ is even, then the index of $E_a$ is smaller than the index of $E_b$.
\end{enumerate}
\item For any integer $t$ with $t>s$, whenever one of the two aggregations $A,B$ has an entry of size $t$, so does the other, and the indices of the two entries are the same. 
\end{enumerate}

It is not hard to see (left to the reader) that 
\begin{equation}\label{ff1}
\mbox{\em $\prec$ is transitive and irreflexive.}
\end{equation}
 It is also easy to see (again left to the reader) that
\begin{equation}\label{ff2}
\mbox{\em For any $i$ with $(i+ 1)\in I$, we have $A_i\prec A_{i+1}$.}
\end{equation}

Next, for aggregations $A$ and $B$, we write $A\approx B$ iff neither $A\prec B$ nor $B\prec A$. In other words, $A\approx B$ holds iff, for any positive integer $t$, whenever one of the two aggregations has an entry of size $t$, so does the other, and the indices of the two entries are the same. Obviously $\approx$ is an equivalence relation. In view of (\ref{bodycond}), 
 it is also clear that 
\begin{equation}\label{ff3}
\mbox{\em $\approx$ partitions the set of all possible aggregations into (only) finitely many equivalence classes.} 
\end{equation}

Now, for a contradiction, assume $I$ is infinite. Then, by (\ref{ff2}), we have an infinite chain \mbox{$A_1\prec A_2\prec A_3\prec\cdots$.} In view of (\ref{ff1}), all aggregations of this chain belong to different $\approx$-equivalence classes, meaning that there are infinitely many such classes. However, (\ref{ff3}) tells us that this is not the case. 
 \end{proof}

We say that two bodies are {\bf consistent} with each other iff one is an extension of the other. This, of course, includes the case of their being simply equal. 

\begin{lemma}\label{jinan}
Consider any $n\in\{0,\ldots,k\}$ and any $i,j\in I$.  Suppose $A_i$ has an entry $[n,B_i]$, and $A_j$ has an entry $[n,B_j]$.  Then $B_i$ and $B_j$ are consistent with each other. 
\end{lemma}

\begin{proof} Assume the conditions of the lemma. Note that $i,j>1$, because $A_1$ has no entries. The case $i=j$ is trivial, so we shall assume that $i<j$. 
The case of either $B_i$ or $B_j$ being empty is also trivial, because the empty body is consistent with every body. Thus, we shall assume that $B_i$ looks like $(\vec{\alpha}_1,\ldots,\vec{\alpha}_a,\vec{\alpha}_{a+1})$ and $B_j$ looks like $(\vec{\beta}_1,\ldots,\vec{\beta}_b,\vec{\beta}_{b+1})$ for some $a,b\geq 0$. 

We prove the lemma by complete  induction on $i+ j$. Assume the aggregation $A_{i-1}$ contains the entry  $[n,B_i]$. Since $(i-1)+j<i+j$, the induction hypothesis applies, according to which  $B_i$ is consistent with $B_j$, as desired. The case of $A_{j-1}$ containing  the entry  $[n,B_j]$ is  similar. Now, for the rest of the present proof, we assume that 
\begin{equation}\label{hunqian}
\mbox{\em $A_{i-1}$ does not have the entry $[n,B_i]$, and  $A_{j-1}$ does not have the entry $[n,B_j]$.}
\end{equation}
Assume $a<b$. Note that then $b\geq 1$. In view of this fact and (\ref{hunqian}), it is easy to see that $A_{j-1}$ contains an $n$-indexed entry whose body is $(\vec{\beta}_1,\ldots,\vec{\beta}_b)$. By the induction hypothesis, $(\vec{\alpha}_1,\ldots,\vec{\alpha}_a,\vec{\alpha}_{a+1})$ is consistent with $(\vec{\beta}_1,\ldots,\vec{\beta}_b)$, meaning (as $a+1\leq b$) that the latter is an extension of the former. Hence, $(\vec{\alpha}_1,\ldots,\vec{\alpha}_a,\vec{\alpha}_{a+1})$ is also consistent with $(\vec{\beta}_1,\ldots,\vec{\beta}_{b+1})$, as desired. The case of $b<a$ will be handled in a similar way. Thus, for the rest of this proof, we further assume that $a=b$.  

Next we claim that 
\begin{equation}\label{hun}
(\vec{\alpha}_1,\ldots,\vec{\alpha}_a)=(\vec{\beta}_1,\ldots,\vec{\beta}_b).
\end{equation}
Indeed, the case of $a,b=0$ is trivial. Otherwise, if $a,b\not=0$, in view of (\ref{hunqian}),  obviously   $A_{i-1}$ contains the entry $[n,(\vec{\alpha}_1,\ldots,\vec{\alpha}_a)]$ and  $A_{j-1}$ contains the entry $[n,(\vec{\beta}_1,\ldots,\vec{\beta}_b)]$. Hence, by the induction hypothesis, the two bodies $(\vec{\alpha}_1,\ldots,\vec{\alpha}_a)$ and $(\vec{\beta}_1,\ldots,\vec{\beta}_b)$ are consistent, which, as $a=b$, simply means that they are the same. (\ref{hun}) is thus verified. 
In view of (\ref{hun}), all that now remains to show is that $\vec{\alpha}_{a+1}=\vec{\beta}_{a+1}$ ($=\vec{\beta}_{b+1}$). 

Assume $a$ is even. Analyzing the work of $\mathbb{MAIN}$ and keeping (\ref{hunqian}) in mind, we see that the activity triple of $A_{i-1}$ must be $(n,C,(\vec{\alpha}_1,\ldots,\vec{\alpha}_a))$ for a certain body $C$, and (thus) $
\mathbb{SIM}_{n}(C^{\mbox{\em odd}},(\vec{\alpha}_1,\ldots,\vec{\alpha}_a)^{\mbox{\em even}})=+\vec{\alpha}_{a+1}$. 
Similarly, the activity triple of $A_{j-1}$ is $(n,D,(\vec{\beta}_1,\ldots,\vec{\beta}_a))$ --- which, by (\ref{hun}), is the same as $(n,D,(\vec{\alpha}_1,\ldots,\vec{\alpha}_a))$ --- for a certain body $D$, and 
$\mathbb{SIM}_{n}(D^{\mbox{\em odd}},(\vec{\alpha}_1,\ldots,\vec{\alpha}_a)^{\mbox{\em even}})=+\vec{\beta}_{a+1}$. Here, if $n=0$, both $C$ and $D$ are $(\seq{})$  and hence consistent with each other. Otherwise, if $n>0$, $A_{i-1}$ contains the entry $[n-1,C]$, and $A_{j-1}$ contains the entry $[n-1,D]$. Then, by the induction hypothesis, again, $C$ is consistent with $D$. Thus, in either case, $C$ and $D$ are consistent. Then clause 1 of Lemma \ref{golemma} implies that 
$\mathbb{SIM}_{n}(C^{\mbox{\em odd}},(\vec{\alpha}_1,\ldots,\vec{\alpha}_a)^{\mbox{\em even}})=\mathbb{SIM}_{n}(D^{\mbox{\em odd}},(\vec{\alpha}_1,\ldots,\vec{\alpha}_a)^{\mbox{\em even}})$. In other words,  $\vec{\alpha}_{a+1}=\vec{\beta}_{a+1}$, as desired. 

The case of $a$ being odd is rather similar. In this case, the activity triple of $A_{i-1}$ is $(n+1,(\vec{\alpha}_1,\ldots,\vec{\alpha}_a),C)$ for a certain body $C$, with $\mathbb{SIM}_{n+1}((\vec{\alpha}_1,\ldots,\vec{\alpha}_a)^{\mbox{\em odd}},C^{\mbox{\em even}})=-\vec{\alpha}_{a+1}$. 
And the activity triple of $A_{j-1}$ is $(n+1,(\vec{\alpha}_1,\ldots,\vec{\alpha}_a),D)$ for a certain body $D$, with  
$\mathbb{SIM}_{n+1}((\vec{\alpha}_1,\ldots,\vec{\alpha}_a)^{\mbox{\em odd}},D^{\mbox{\em even}})=-\vec{\beta}_{a+1}$. If either $C$ or $D$ is empty, then 
the two bodies are consistent with each other. Otherwise, if both $C$ and $D$ are nonempty, then $A_{i-1}$ contains the entry $[n+1,C]$, $A_{j-1}$ contains the entry $[n+1,D]$ and hence, by the induction hypothesis, $C$ and $D$ are again consistent. Then clause 2 of Lemma \ref{golemma} implies that 
$\mathbb{SIM}_{n+1}((\vec{\alpha}_1,\ldots,\vec{\alpha}_a)^{\mbox{\em odd}},C^{\mbox{\em even}})=\mathbb{SIM}_{n+1}((\vec{\alpha}_1,\ldots,\vec{\alpha}_a)^{\mbox{\em odd}},D^{\mbox{\em even}})$, meaning that  $\vec{\alpha}_{a+1}=\vec{\beta}_{a+1}$, as desired.   
\end{proof}

For each $n$ with $0\leq n\leq k$, we define  $\mathbb{B}_n$, called the {\bf ultimate body} for $n$,   
as the smallest (smallest-size) body such that, for every $i\in I$, whenever 
$A_i$ has an $n$-indexed entry, $\mathbb{B}_n$ is an extension of the body of that entry. In view of Lemma \ref{jinan}, such a $\mathbb{B}_n$ always exists. 

When $\vec{\alpha}=\seq{\alpha_1,\ldots,\alpha_a}$ is a movesequence and $\xx$ is one of the players $\pp$ or $\oo$, we shall write $\xx\vec{\alpha}$ for the run  $\seq{\xx\alpha_1,\ldots,\xx\alpha_a}$. 

Consider any $n\in\{0,\ldots,k\}$, and let $(\vec{\alpha}_1,\ldots,\vec{\alpha}_s)$ be the ultimate body $\mathbb{B}_n$ for $n$.  We define  $\mathbb{R}_n$, called the {\bf ultimate run} for $n$,  as   the run $\seq{\pp\vec{\alpha}_1,\oo\vec{\alpha}_2,\ldots}$ obtained from $\mathbb{B}_n$ by replacing each $\vec{\alpha}_i$ with $\pp\vec{\alpha}_i$ if $i$ is odd, and with $\oo\vec{\alpha}_i$ if $i$ is even.

Some more notation and terminology. When $\Gamma$ and $\Delta$ are runs, we write $\Gamma\preceq \Delta$ to mean that $\Gamma$ is a (not necessarily proper) initial segment of $\Delta$. Next, as always in CoL, $\gneg \Gamma$ means the result of changing in $\Gamma$ each label $\pp$ to $\oo$ and vice versa.  $\Gamma^{0.}$  means the result of deleting from $\Gamma$ all moves (together with their labels, of course) except those of the form $0.\alpha$, and then further deleting the prefix ``$0.$''  in the remaining moves. Similarly for $\Gamma^{1.}$. Intuitively, when $\Gamma$ is a play of a parallel disjunction $G_0\mld G_1$ of games,  $\Gamma^{0.}$ (resp. $\Gamma^{1.}$) is the play that has taken place --- according to the scenario of $\Gamma$ --- in the $G_0$ (resp. $G_1$) component. We say that $\Gamma$ is {\bf bipartite} iff every move of $\Gamma$ has one of the two prefixes ``$0.$'' or ``$1.$''.  Obviously being bipartite is a necessary (but not sufficient) condition for $\Gamma$ to be a legal run of   $G_0\mld G_1$.

\begin{lemma}\label{wuhan} \ 
\begin{description}  
\item[1.] (a) There is a run $\Gamma_0$ generated by ${\cal H}_0$ such that $\mathbb{R}_0\preceq \Gamma_0$. (b) Furthermore, 
if $\mathbb{R}_{0}$ is a legal and reasonable run of $F'(0)$, then we simply have $\mathbb{R}_0= \Gamma_0$. 
\item[2.]  (a) For every $n\in\{1,\ldots,k\}$, there is a bipartite run $\Gamma_n$ generated by ${\cal H}_n$ such that  $\gneg\mathbb{R}_{n-1}\preceq \Gamma^{0.}_{n}$ and $\mathbb{R}_n\preceq \Gamma^{1.}_{n}$. (b) Furthermore, 
if $\mathbb{R}_{n-1}$ and  $\mathbb{R}_{n}$ are legal and reasonable runs of $F'(n-1)$ and $F'(n)$, respectively, then we simply have  $\gneg\mathbb{R}_{n-1}= \Gamma^{0.}_{n}$ and $\mathbb{R}_n= \Gamma^{1.}_{n}$. 
\end{description} 
\end{lemma}

\begin{proof} We first verify that 
\begin{equation}\label{hhh}
\mbox{\em The size of $\mathbb{B}_k$ is odd.}
\end{equation}
Indeed, let $g$ be the greatest element of $I$, which exists by Lemma \ref{beijing}. Consider the last, i.e. $g$th, iteration of $\mathbb{MAIN}$. Obviously the aggregation $A=A_g$ dealt with throughout  that iteration is  passive, for otherwise there would be a next iteration. This, by the definition of ``passive'', means that the last --- $k$-indexed --- entry of $A_g$ is odd-size. With some analysis of  Case 2 of $\mathbb{MAIN}$ and keeping Lemma \ref{jinan} in mind,  it can be seen that the body of that entry is $\mathbb{B}_k$, for otherwise, again, there would be a next iteration. 

Next, we verify clause 2 of the lemma (skipping clause 1 for now) simultaneously with the following claim:
\begin{equation}\label{hhhh}
\mbox{\em For every $n\in\{1,\ldots,k\}$, the size of $\mathbb{B}_{n-1}$ is odd.}
\end{equation}

Consider any  $n\in\{1,\ldots,k\}$ and the ultimate bodies $\mathbb{B}_{n-1}=(\vec{\alpha}_1,\ldots,\vec{\alpha}_s)$ and  
$\mathbb{B}_{n}=(\vec{\beta}_1,\ldots,\vec{\beta}_t)$.   We want to show that the size of $\mathbb{B}_{n-1}$ is odd, and that clause 2 of the lemma holds for $n$. Our proof of (\ref{hhhh}) is, in fact, by induction on $k-n$. By (\ref{hhh}) if $k=n$ (i.e., if we are dealing with the basis of induction), and by the induction hypothesis if $k<n$ (i.e., if we are dealing with the inductive step), 
 we have: 
\begin{equation}\label{jjj}
\mbox{\em $t$, i.e., the size of $\mathbb{B}_n$, is odd.}
\end{equation}

 Let $i$ be the smallest member of $I$ such that $A_i$ 
contains the entry $[n,\mathbb{B}_n]$. Obviously (\ref{jjj}) implies  that ($i> 1$ and) $A_{i-1}$ has the entry $[n-1,C]$ for 
a certain odd-size body $C$ and, with $(n,C,(\vec{\beta}_1,\ldots,\vec{\beta}_{t-1}))$ being the activity triple of 
$A_{i-1}$,  $\mathbb{SIM}_n(C^{\mbox{\em odd}},(\vec{\beta}_1,\ldots,\vec{\beta}_{t-1})^{\mbox{\em even}})=+\vec{\beta}_t$. By  Lemma \ref{jinan}, $\mathbb{B}_{n-1}$ is consistent with $C$, which, in view of $\mathbb{B}_{n-1}$'s being ultimate, means that 
$\mathbb{B}_{n-1}$ is an extension of $C$. Hence, by clause 1 of Lemma \ref{golemma}, 
\begin{equation}\label{chenghai} 
\mathbb{SIM}_n((\vec{\alpha}_1,\ldots,\vec{\alpha}_s)^{\mbox{\em odd}},(\vec{\beta}_1,\ldots,\vec{\beta}_{t-1})^{\mbox{\em even}})=+\vec{\beta}_t.
\end{equation}
  
Also, note that, since $\mathbb{B}_{n-1}=(\vec{\alpha}_1,\ldots,\vec{\alpha}_s)$ is an extension of the odd-size $C$, we have $s\geq 1$. 

We now claim the following: 
\begin{equation}\label{chenghaii} 
\mbox{\em For no  $r$ with $1\leq r< s$ do we have $\mathbb{SIM}_n((\vec{\alpha}_1,\ldots,\vec{\alpha}_r)^{\mbox{odd}},(\vec{\beta}_1,\ldots,\vec{\beta}_{t-1})^{\mbox{even}})=+\vec{\beta}_t$.}
\end{equation}

Indeed, for a contradiction, assume $r$ is the smallest number with $1\leq r< s$ such that 
\begin{equation}\label{55}
\mathbb{SIM}_n((\vec{\alpha}_1,\ldots,\vec{\alpha}_r)^{\mbox{\em odd}},(\vec{\beta}_1,\ldots,\vec{\beta}_{t-1})^{\mbox{\em even}})=+\vec{\beta}_t.
\end{equation}
 Note that $r$ has to be odd, for otherwise $(\vec{\alpha}_1,\ldots,\vec{\alpha}_r)^{\mbox{\em odd}}=(\vec{\alpha}_1,\ldots,\vec{\alpha}_{r-1})^{\mbox{\em odd}}$ and hence $r-1$ would be a number with $1\leq r-1\leq s$ smaller than $r$ satisfying (\ref{55}).  
Let $j$ be the smallest element of $I$ such that $A_j$ contains the entry $[n-1,(\vec{\alpha}_1,\ldots,\vec{\alpha}_{r+1})]$. The size of this entry is even (because $r$ is odd) and non-zero.   Hence obviously ($j>1$ and) $A_{j-1}$ contains the entry    $[n-1,(\vec{\alpha}_1,\ldots,\vec{\alpha}_{r})]$, and we have 
$\mathbb{SIM}_{n}((\vec{\alpha}_1,\ldots,\vec{\alpha}_{r})^{\mbox{\em odd}},D^{\mbox{\em even}})=-\vec{\alpha}_{r+1}$, where $D$ is either empty or else $A_{j-1}$ contains an $n$-indexed, even-size entry and $D$ is the body of that entry. In either case (in the latter case by Lemma \ref{jinan}), $(\vec{\beta}_1,\ldots,\vec{\beta}_{t-1})$ is an extension of $D$.  Therefore, by clause 2 of Lemma \ref{golemma},  \[\mathbb{SIM}_n((\vec{\alpha}_1,\ldots,\vec{\alpha}_{r})^{\mbox{\em odd}},(\vec{\beta}_1,\ldots,\vec{\beta}_{t-1})^{\mbox{\em even}})=-\vec{\alpha}_{r+1}.\] The above, however, contradicts (\ref{55}). Claim (\ref{chenghaii}) is thus proven. 

We can now see why (\ref{hhhh}) holds for $n$, i.e., why the size of $\mathbb{B}_{n-1}=(\vec{\alpha}_1,\ldots,\vec{\alpha}_s)$ is odd. Indeed, if $s$ was even, then, with the earlier observed fact $s\geq 1$ in mind,  we would have $(\vec{\alpha}_1,\ldots,\vec{\alpha}_{s})^{\mbox{\em odd}}= (\vec{\alpha}_1,\ldots,\vec{\alpha}_{s-1})^{\mbox{\em odd}}$ and hence, by 
(\ref{chenghai}), 
\[\mathbb{SIM}_n((\vec{\alpha}_1,\ldots,\vec{\alpha}_{s-1})^{\mbox{\em odd}},(\vec{\beta}_1,\ldots,\vec{\beta}_{t-1})^{\mbox{\em even}})=+\vec{\beta}_t;\]
claim (\ref{chenghaii}), however, tells us that the above is impossible. This completes our inductive step for (\ref{hhhh}).

Let us now remember the definition of $\mathbb{SIM}_n$ and imagine how (\ref{chenghai}) is computed. Let $+\vec{\psi}_1,\ldots,+\vec{\psi}_d$ be the positive  values that the variable $S$ of the procedure $\mathbb{SIM}_n$ goes through when computing (\ref{chenghai}), and let 
$-\vec{\nu}_1,\ldots,-\vec{\nu}_c$ be the negative values. 
It is clear that $d$ equals the size of $(\vec{\beta}_1,\ldots,\vec{\beta}_{t-1})^{\mbox{\em even}}$ plus one.  In view of (\ref{chenghaii}), it is also not hard to see that $c$ equals the size of $(\vec{\alpha}_1,\ldots,\vec{\alpha}_s)^{\mbox{\em odd}}$   minus one. Since --- as we already know 
--- both $t$ and $s$ are odd, the above means that 

\marginpar{vvv}
\begin{equation}\label{vvv}
\mbox{\em $c=(s-1)/2$ \ and \ $d=(t+1)/2$.}
\end{equation} 

We  now claim that 
\begin{equation}\label{cc1}
(\vec{\psi}_1,\ldots,\vec{\psi}_d)=(\vec{\beta}_1,\ldots,\vec{\beta}_t)^{\mbox{\em odd}}.
\end{equation}
In view of (\ref{vvv}),  the above claim simply means that we have $\vec{\psi}_{(b+1)/2}=\vec{\beta}_{b}$ for each odd member $b$ of $\{1,\ldots,t\}$. But  indeed, consider any such $b$. Let $A_i$ be the smallest number in $I$ such that $A_i$ contains the entry $[n,(\vec{\beta}_1,\ldots,\vec{\beta}_b)]$. Then, since $b$ is odd, $A_{i-1}$ obviously contains the entry  $[n-1,D]$ for a certain odd-size body $D$ such that,
with $(n,D, (\vec{\beta}_1,\ldots,\vec{\beta}_{b-1}))$ being the activity triple of $A_{i-1}$, we have
 $\mathbb{SIM}_n(D^{\mbox{\em odd}},(\vec{\beta}_1,\ldots,\vec{\beta}_{b-1})^{\mbox{\em even}})=+\vec{\beta}_b$. In view of Lemma \ref{jinan},  $(\vec{\alpha}_1,\ldots,\vec{\alpha}_{s})$ is an extension of $D$. Hence, by clause 1 of Lemma \ref{golemma}, 
\begin{equation}\label{jul22a}
\mathbb{SIM}_n((\vec{\alpha}_1,\ldots,\vec{\alpha}_s)^{\mbox{\em odd}},(\vec{\beta}_1,\ldots,\vec{\beta}_{b-1})^{\mbox{\em even}})=+\vec{\beta}_b.
\end{equation} 
But how is the computation of (\ref{jul22a}) different from the computation of (\ref{chenghai})? The two computations 
obviously proceed in exactly the same ways, with the variable $S$ of $\mathbb{SIM}_n$ going through exactly the same values, 
with the only difference that, while the computation of (\ref{jul22a}) stops after $S$ takes its $((b+1)/2)$th positive value 
$+\vec{\psi}_{(b+1)/2}$ and returns that value as $+\vec{\beta}_b$, the computation of (\ref{chenghai}) continues its work further (unless $b=t$) 
until the value of $S$ becomes $+\vec{\psi}_{(t+1)/2}$. As we see, $\vec{\psi}_{(b+1)/2}$ is indeed the same 
as $\vec{\beta}_b$. 

Next we claim that 
\begin{equation}\label{cc2}
(\vec{\nu}_1,\ldots,\vec{\nu}_c)=(\vec{\alpha}_1,\ldots,\vec{\alpha}_s)^{\mbox{\em even}}.
\end{equation}
Our argument here is very similar to the preceding one.  In view of (\ref{vvv}),  the above claim simply means that we have $\vec{\nu}_{b/2}=\vec{\alpha}_{b}$ for each even member $b$ of $\{1,\ldots,s\}$. But  indeed, consider any such $b$.
 Let $A_i$ be the smallest number in $I$ such that $A_i$ contains the entry $[n-1,(\vec{\alpha}_1,\ldots,\vec{\alpha}_b)]$. Then, since $b$ is even and nonzero, $A_{i-1}$ obviously contains the entry $[n-1,(\vec{\alpha}_1,\ldots,\vec{\alpha}_{b-1})]$, and we have  $\mathbb{SIM}_n((\vec{\alpha}_1,\ldots,\vec{\alpha}_{b-1})^{\mbox{\em odd}},D^{\mbox{\em even}})=-\vec{\alpha}_b$, where $D$ is an even-size body such that $D$ is either empty or else $A_{i-1}$ has the entry $[n,D]$.  In view of Lemma \ref{jinan},  $(\vec{\beta}_1,\ldots,\vec{\beta}_{t})$ is an extension of $D$. But (\ref{jjj}) implies that $(\vec{\beta}_1,\ldots,\vec{\beta}_{t})^{\mbox{\em even}}=(\vec{\beta}_1,\ldots,\vec{\beta}_{t-1})^{\mbox{\em even}}$. Thus, $(\vec{\beta}_1,\ldots,\vec{\beta}_{t-1})^{\mbox{\em even}}$ is an extension of $D^{\mbox{\em even}}$. 
Hence, by clause 2 of Lemma \ref{golemma},  
\begin{equation}\label{jul22b}
\mathbb{SIM}_n((\vec{\alpha}_1,\ldots,\vec{\alpha}_{b-1})^{\mbox{\em odd}},(\vec{\beta}_1,\ldots,\vec{\beta}_{t-1})^{\mbox{\em even}})=-\vec{\alpha}_b.
\end{equation} But how is the computation of (\ref{jul22b}) different from the computation of (\ref{chenghai})? Again, they proceed in exactly the same ways, with the variable $S$ going through exactly the same values, with the only difference that, while the computation of (\ref{jul22b}) stops after $S$ takes its $(b/2)$th negative value $-\vec{\nu}_{b/2}$ and returns that value as $-\vec{\alpha}_b$, the computation of (\ref{chenghai}) continues its work further, with   (\ref{vvv}) guaranteeing  that  the latter does not stop ``too early'', i.e., before the value of its $S$ becomes $-\vec{\nu}_{b/2}$.
 As we see, $\vec{\nu}_{b/2}$ is indeed the same as $\vec{\alpha}_b$.

Having verified (\ref{cc1}) and (\ref{cc2}), let us imagine the computation of (\ref{chenghai}) once again. Let $a$ be the size of $(\vec{\alpha}_1,\ldots,\vec{\alpha}_{s})^{\mbox{\em odd}}$ and $b$ the size of $(\vec{\beta}_1,\ldots,\vec{\beta}_{t-1})^{\mbox{\em even}}$. 
Obviously what the procedure $\mathbb{SIM}_n$ does when computing (\ref{chenghai}) is that it traces a certain computation branch $B$ of  ${\cal H}_n$ --- more precisely, only the first $(a+b)\mathfrak{L}$ steps of $B$, where the adversary's last chunk of moves (either $\vec{\alpha}_s$ or $\vec{\beta}_{t-1}$) occurred on step $(a+b-1)\mathfrak{L}$, and where the adversary never ever made any subsequent moves. Let $\Gamma_n$ be the run spelled by $B$. We may assume that ${\cal H}_n$ is well-behaved enough to never make ``pathologically wrong'' and pointless moves that do not have one of the two prefixes ``$0.$'' or ``$1.$'', as we did regarding  another similar ``well-behavedness'' condition in the proof of Lemma \ref{golemma}. On this assumption, it is obvious that $\Gamma_n$ is bipartite.  If we only look at the first  $(a+b)\mathfrak{L}$ steps of $B$ and the initial segment $\Delta_n$ of $\Gamma_n$ consisting of the moves made by the two players during those steps, with a little thought and with  (\ref{vvv}) in mind,
 we find that 
\begin{equation}\label{www}  
\Delta_{n}^{1.}=\seq{\pp\vec{\psi}_1,\oo\vec{\beta}_2,\pp\vec{\psi}_2,\oo\vec{\beta}_4,\pp\vec{\psi}_3,\oo\vec{\beta}_6,\ldots,\pp\vec{\psi}_{d-1},\oo\vec{\beta}_{t-1},\pp\vec{\psi}_d}
\end{equation}
and 
\begin{equation}\label{uuu}  
\Delta_{n}^{0.}=\seq{\oo\vec{\alpha}_1,\pp\vec{\nu}_1,\oo\vec{\alpha}_3,\pp\vec{\nu}_2,\oo\vec{\alpha}_5,\pp\vec{\nu}_3,\ldots,\oo\vec{\alpha}_{s-2},\pp\vec{\nu}_c,\oo\vec{\alpha}_s}.
\end{equation}

Now, remembering our definition of ``ultimate run'', (\ref{cc1}) and (\ref{www}) together mean nothing but that $\Delta_{n}^{1.}=\mathbb{R}_n$. Similarly, (\ref{cc2}) and (\ref{uuu}) together mean that $\Delta_{n}^{0.}=\gneg \mathbb{R}_{n-1}$. This takes care of subclause (a) of clause 2 of the present lemma because, as we remember, $\Delta_n\preceq \Gamma_n$. 

For subclause (b) of clause 2, assume $\mathbb{R}_{n-1}$ and $\mathbb{R}_{n}$ are legal and reasonable. Remember that $\Delta_n$ consists of the moves made by the two players during the first $(a+b)\mathfrak{L}$ steps of the computation branch $B$ that spells $\Gamma_n$, and that, in that branch, ${\cal H}_n$'s adversary never moved after step $(a+b-1)\mathfrak{L}$. If so, Lemma \ref{m29a} tells us that ${\cal H}_n$ would never move after step   $(a+b)\mathfrak{L}$. This means that $\Delta_n$ and $\Gamma_n$ are simply the same, and thus, in view of what we already know about $\Delta_n$, we have $\Gamma_{n}^{1.}=\mathbb{R}_n$ and  $\Gamma_{n}^{0.}=\gneg \mathbb{R}_{n-1}$, as desired. Clause 2 of the lemma is now fully verified. 

Clause 1 of the lemma can be verified in a similar but considerably easier way, relying on the fact that, by (\ref{hhhh}), the size of $\mathbb{B}_0$ is odd. Such a verification --- if necessary --- is left as an exercise for the reader. 
\end{proof}

\begin{lemma}\label{shan}
For every $n\in\{0,\ldots,k\}$, $\mathbb{R}_n$ is a legal and reasonable run of $F'(n)$.
\end{lemma}

\begin{proof} Remember the notion of a $\xx$-illegal run from Definition 3.1 of \cite{cla4}. As in the proofs of Lemmas \ref{golemma} and \ref{wuhan},  we can safely make yet another ``well-behavedness'' assumption regarding  the machine(s) that we are dealing with, according to which, 
 for any $n\in\{1,\ldots,k\}$, whenever 
$\Gamma$ is a run generated by ${\cal H}_n$, $\Gamma^{0.}$ is not a $\pp$-illegal run of $\gneg F'(n-1)$ and $\Gamma^{1.}$ is not a $\pp$-illegal run of $F'(n)$. In intuitive terms, the present assumption simply means that ${\cal H}_n$ plays legally not only the overall game $\gneg F'(n-1)\mld F'(n)$, but also in both individual disjuncts of it, even if its adversary has already made an illegal move outside that disjunct. 

Also, note that our assumption regarding ${\cal H}_n$'s playing $\gneg F'(n-1)\mld F'(n)$ reasonably automatically extends from the overall game to the individual disjuncts of it.  

First, consider the case  $n=0$.  Assume $\mathbb{R}_0$ is an illegal or unreasonable run of $F'(0)$. By clause 1(a) of Lemma \ref{wuhan},  $\mathbb{R}_0$ is an initial segment of a certain run $\Gamma_0$ generated by ${\cal H}_0$. Therefore, in view of our assumption that ${\cal H}_0$ plays $F'(0)$ legally and reasonably, the only way $\mathbb{R}_0$ could be illegal or unreasonable is if ${\cal H}_0$'s adversary made  an illegal or unreasonable move in it. This is however impossible, because then $\gneg\mathbb{R}_0$ would be an illegal or unreasonable run of $\gneg F'(0)$, with player $\pp$ 
 being responsible for making it so. But, according to clause 2(a) of Lemma \ref{wuhan}, a certain extension $\Gamma^{0.}_{1}$ of $\gneg \mathbb{R}_0$ is a run generated by ${\cal H}_1$ (with ${\cal H}_1$ playing as $\pp$) in the component $\gneg F'(0)$ of $\gneg F'(0)\mld F'(1)$, which, of course, is illegal or unreasonable because its initial segment $\gneg \mathbb{R}_0$ is so. Thus, ${\cal H}_1$ has played illegally or unreasonably in the $\gneg F'(0)$ component. This, however, is impossible in view of our assumption (see the preceding two paragraphs) that ${\cal H}_1$ plays in the  $\gneg F'(0)$ component of $\gneg F'(n-1)\mld F'(n)$ legally and reasonably. 

The case $0< n< k$ is handled in a similar way, focusing only on the consequent $F'(n)$ of $F'(n-1)\mli F'(n)$ and the corresponding runs $\mathbb{R}_n$ and $\Gamma_{n+1}^{0.}$.  

Finally, consider the case $n=k$. Just as in the preceding cases, ${\cal H}_k$ cannot be responsible for making $\mathbb{R}_k$ an illegal or unreasonable run of $F'(k)$. Analyzing Case 2 of the description of $\mathbb{MAIN}$, it is rather clear that ${\cal H}_k$'s imaginary environment does not make $\mathbb{R}_k$ illegal or unreasonable either. This is so because $\mathbb{MAIN}$ simply stops any activities --- including copying ${\cal M}_k$'s adversary's moves and adding them to $\mathbb{B}_k$ --- once it detects an illegal move by the environment; also, $\mathbb{MAIN}$ moderates ${\cal M}_k$'s environment's unreasonable moves by  ``making them reasonable'' before copying and adding them to $\mathbb{B}_k$. 
\end{proof}

\begin{lemma}\label{shantou}
For every $n\in\{0,\ldots,k\}$, $\mathbb{R}_n$ is  $\pp$-won run of $F'(n)$.
\end{lemma}

\begin{proof} Induction on $n$. According to clause 1(b) of Lemma \ref{wuhan}, in conjunction with Lemma \ref{shan}, $\mathbb{R}_0$ is a run generated by ${\cal H}_0$. So, since ${\cal H}_0$ wins $F'(0)$, $\mathbb{R}_0$ is a $\pp$-won run of $F'(0)$.

Next, consider any $n$ with $0< n\leq k$. According to clause 2(b) of Lemma \ref{wuhan}, in conjunction with Lemma \ref{shan}, there is a bipartite run $\Gamma_n$ generated by ${\cal H}_n$ such that $\Gamma^{0.}_{n}=\gneg \mathbb{R}_{n-1}$ and $\Gamma^{1.}_{n}=\mathbb{R}_n$. But we know that ${\cal H}_n$ wins $\gneg F'(n-1)\mld F'(n)$. So, $\Gamma_n$ has to be a $\pp$-won run of  $\gneg F'(n-1)\mld F'(n)$, meaning that either $\Gamma^{0.}_{n}$, i.e.   $\gneg \mathbb{R}_{n-1}$, is a $\pp$-won run of $\gneg F'(n-1)$, or $\Gamma^{1.}_{n}$, i.e. $\mathbb{R}_{n}$, is a $\pp$-won run of $F'(n)$.  But, by the induction hypothesis, $\mathbb{R}_{n-1}$ is a $\pp$-won run of $F'(n-1)$. This obviously means that  $\gneg \mathbb{R}_{n-1}$ is a $\oo$-won (and thus not $\pp$-won) run of $\gneg F'(n-1)$. Therefore, $\mathbb{R}_{n}$ is a $\pp$-won run of $F'(n)$. 
\end{proof}

A run $\Pi$ is said to be a {\bf $\pp$-delay} of a run $\Sigma$ iff: (1)  for either player $\xx\in\{\twg,\tlg\}$, the subsequence of the $\xx$-labeled moves of $\Pi$ is the same as that of $\Sigma$, and (2)
for any $x,y\geq 1$, if the $x$th $\oo$-labeled move is made earlier than (is to the left of) the $y$th $\pp$-labeled move in $\Sigma$, then so is it in $\Pi$. 
For instance, $\seq{\oo\theta_1,\pp\rho_1,\oo\theta_2,\pp\rho_2}$ is a $\pp$-delay of  $\seq{\pp\rho_1,\oo\theta_1,\pp\rho_2,\oo\theta_2}$. 
It is rather obvious that, whenever $\Sigma$ is a $\pp$-won run of $F'(k)$ and $\Pi$ is a $\pp$-delay of $\Sigma$,  $\Pi$ is also a $\pp$-won run of $F'(k)$.\footnote{The same, of course, holds for any sentence in the role of $F'(k)$. Furthermore, according to a known fact of CoL (\cite{Jap03,Japfin}), the same simply holds for any static game.} 
  From Lemma \ref{shantou}, we know that $\mathbb{R}_k$ is a $\pp$-won run of $F'(k)$. Therefore we have:
\begin{equation}\label{trr}
\mbox{\em Whenever a run $\Pi$ is a $\pp$-delay of $\mathbb{R}_k$, $\Pi$ is a $\pp$-won run of $F'(k)$.}
\end{equation}

Let $\Upsilon$ be the run generated by ${\cal M}_k$ that took place in the real play of $F'(k)$. How does $\Upsilon$ relate to $\mathbb{R}_k$? As promised earlier, the real play of $F'(k)$ --- that is, the run $\Upsilon$ --- would be ``essentially synchronized'' with the play $\mathbb{R}_k$ by ${\cal H}_k$ in the consequent of $F'(k-1)\mli F'(k)$, meaning that $\Upsilon$ is ``essentially the same'' as $\mathbb{R}_k$. The qualification ``essentially'' implies that the two runs, while being similar, may not necessarily be strictly identical. 

One reason why $\mathbb{R}_k$ and $\Upsilon$ may not be exactly the same is that, as can be seen from Case 2 of the description of $\mathbb{MAIN}$, if $\Upsilon$ contains an illegal (with respect to $F'(k)$) move by $\oo$, such a move does not appear in $\mathbb{R}_k$. However, if the adversary made an illegal move in the (real) play of $F'(k)$, ${\cal M}_k$ is an automatic winner. So, we can and will safely assume that $\Upsilon$ does not contain illegal moves by ${\cal M}_k$'s adversary, for otherwise the case is trivial.

Another reason why $\mathbb{R}_k$ may differ from $\Upsilon$ is that, again as seen from Case 2 of the description of $\mathbb{MAIN}$, if $\Upsilon$ contains an unreasonable (with respect to $F'(k)$) move by $\oo$, such a move appears in $\mathbb{R}_k$ in a ``moderated'' and hence altered form. Namely, if ${\cal H}_k$'s adversary chose some ``unreasonably long'' constant $a$ for $z$ in a subcomponent $\ada z G$ of $F'(k)$, then the same move will appear in $\mathbb{R}_k$ as if $0$ was chosen instead of $a$. Note, however, that having made the above unreasonable choice makes $\oo$ lose in the $\ada z G$ component. So, ``moderating'' $\oo$'s unreasonable moves can only increase rather than decrease $\oo$'s chances to win the overall game. That is, if $\pp$ (i.e. ${\cal M}_k$) wins the game even after such moderation of the adversary's unreasonable moves,  it would also win (``even more so'') without moderation.  For this reason, we can and will further safely assume that ${\cal M}_k$'s environment plays not only legally, but also reasonably.

But even if ${\cal M}_k$'s adversary has played $\Upsilon$ legally and reasonably, there is one (last) reason remaining that could make  $\mathbb{R}_k$ ``somewhat'' different from $\Upsilon$. Namely, with some thought, one can see that $\Upsilon$ may be a proper $\pp$-delay of (rather than equal to) $\mathbb{R}_k$. Luckily, however, by (\ref{trr}), $\Upsilon$ is still a $\pp$-won run of $F'(k)$.

Thus, as desired, ${\cal M}_k$ wins $F'(k)$, and hence $\cal M$ wins $F(x)$. 

It remains to verify  that ${\cal M}$ runs in polynomial space. Remember from Subsection \ref{sagree}  that $\mathfrak{l}$ is the size of the greatest of the constants chosen by  ${\cal M}$'s environment for the free variables of $F(x)$. This, of course, means that the background of any clock cycle of ${\cal M}_k$ in any scenario of its work 
 will be at least 
$\mathfrak{l}$. For this reason, in order to show that ${\cal M}$ runs in polynomial space, it will be sufficient to show that the spacecost of any clock cycle of ${\cal M}_k$ is bounded by a certain polynomial function of the argument $\mathfrak{l}$. Correspondingly, in what follows, whenever we say ``polynomial'', it is to be understood as ``polynomial in $\mathfrak{l}$''.  

In asymptotic terms, the space consumed by ${\cal M}_k$ --- namely, by any given $i$th ($i\in I$) iteration of $\mathbb{MAIN}$ --- is the sum of the following two components:
\begin{eqnarray}
& \mbox{\em the space needed to hold (the value of) the aggregation $A$;} & \label{spc1}\\
& \mbox{\em the space needed to update $A=A_i$ to $A=A_{i+1}$ (if $(i+ 1)\in I$).}  & \label{spc2}
\end{eqnarray}

In verifying that the quantity (\ref{spc1}) is polynomial, we  observe that, by conditions (iii) and (iv) of Subsection \ref{saggr}, an aggregation cannot contain two same-size entries. Next, by (\ref{bodycond}), the size of an entry never exceeds $2\mathfrak{e}+ 1$. Thus, the number of entries in $A$  is bounded by the constant $2\mathfrak{e}+ 1$. For this reason, it is now sufficient to just show that holding any given entry $[n,B]$ of $A$ only takes a polynomial amount of space. But this is indeed so. The component $n$ can be written in linear space. As for the component $B$, in view of Lemma \ref{jinan}, the amount of space needed to represent it obviously does not exceed the amount of space needed to represent the run $\mathbb{R}_n$. Since, by Lemma \ref{shan}, such a run is legal and reasonable, the overall number of moves in it cannot exceed the constant bound $\mathfrak{d}$ and, as observed in Subsection \ref{sagree}, there is a polynomial function $\eta$ such that the size of no move (due to being reasonable) exceeds $\eta(\mathfrak{l})$. So, (\ref{spc1}) is indeed polynomial.

Verifying the polynomiality of the quantity (\ref{spc2}) means verifying that $\mathbb{SIM}_n$  runs in polynomial space. But this is indeed so. Let us only consider the case $n>0$, with the case $n=0$ being similar but simpler. The space used by $\mathbb{SIM}_n$ is the space needed for maintaining (the contents of) the variables $x,y,z,S,T$,  the space needed for simulating ${\cal H}_n$, and the space needed to keep track of how many steps of ${\cal H}_n$ have been simulated so far. The sizes of $x,y,z$ are bounded by a constant (namely, $|2\mathfrak{e}+1|$). And, asymptotically, the space needed for maintaining  $S$ and $R$ obviously does not exceed the space needed to hold $\mathbb{R}_n$, which, as we observed in the preceding paragraph, is polynomial. So is 
the space needed for simulating ${\cal H}_n$, because, by our assumption, 
 ${\cal H}_n$ runs in polynomial space $\phi$, and  simulating ${\cal H}_n$ obviously takes about the same amount of space. Finally, the count of simulated steps never exceeds a certain constant (specifically, $2\mathfrak{e}+1$) times $\mathfrak{L}$ and, looking back at the definition of $\mathfrak{L}$ in Section \ref{sagree}, it is clear that recording such a quantity can be done with polynomial space --- namely, while $\mathfrak{L}$ is exponential, the size $|\mathfrak{L}|$ of its binary representation is only polynomial.

Thus, as promised, $\cal M$ is indeed a polynomial space solution of $F(x)$. An (optimal) explicit polynomial bound  for the space complexity of $\cal M$ can be written after a long and tedious but otherwise certainly doable analysis of all details of the work of $\cal M$.

\section{The extensional completeness of $\arfive$}\label{sectcompl}
This section is devoted to proving the completeness part of Theorem \ref{tt1}. It means showing that, for any arithmetical problem $A$ that has a polynomial space  solution,  there is a theorem of $\arfive$ which, under the standard interpretation, equals (``expresses'') $A$. 

So, let us pick an arbitrary polynomial-space-solvable arithmetical problem $A$. By definition, $A$ is an arithmetical problem because, for some sentence $X$, $A=X^\dagger$. For the rest of this section, we fix such an  \(X,\)  and fix \({\cal X}\) as an HPM that solves $A$ (and hence $X^\dagger$) in polynomial space. Specifically, we assume that $\cal X$ runs in space 
$\chi$, 
where $\chi$, which we also fix for the rest of this section,  is a single-variable term --- and hence can be seen/written as an explicit polynomial function --- with $\chi(x)\geq x$ for all $x$.  We also agree that, throughout this section, ``{\bf formula}'' exclusively means a subformula of $X$, in which however some variables may  be renamed. 

$X$ may not necessarily be provable in $\arfive$, and our goal is to construct another sentence $\overline{X}$ for which, just like for $X$, we have $A=\overline{X}^\dagger$ and which, perhaps unlike $X$, is  provable in $\arfive$.

Remember the sentence $\mathbb{L}$ from Section 14.3 of \cite{cla4}, saying that $\cal X$ does not win $X$ in time $\chi$. Here we redefine that sentence so that now it says the same but about space rather than time. Namely, let $E_1(\vec{x})\ldots,E_n(\vec{x})$ be all subformulas of $X$, where all free variables of each $E_i(\vec{x})$ are among $\vec{x}$. Then  our present  
$\mathbb{L}$   is the $\mld$-disjunction of  natural formalizations of the following statements: 
\begin{quote} {\em 
\begin{enumerate}
\item There is a $\pp$-illegal position of $X$ spelled on the run tape of $\cal X$ on some clock  cycle of some computation branch of $\cal X$. 
\item There is a clock cycle $c$ in some computation branch of $\cal X$  whose spacecost (see Section \ref{a18}) exceeds   $\chi(\ell)$, where $\ell$ is the background of $c$. 
\item There is a (finite) legal run $\Gamma$ of $X$ generated by $\cal X$ and a tuple $\vec{c}$ of constants ($\vec{c}$ of the same length as $\vec{x}$) such that:
\begin{itemize}
\item  $\seq{\Gamma}X=E_1(\vec{c})$, and we have $\gneg  \elz{E_1(\vec{c})} $ (i.e., $ \elz{E_1(\vec{c})} $ is false),  
\item  or $\ldots$, or 
\item $\seq{\Gamma}X=E_n(\vec{c})$, and we have $\gneg  \elz{E_n(\vec{c})} $ (i.e., $ \elz{E_n(\vec{c})} $ is false). 
\end{itemize}
\end{enumerate} } 
\end{quote}

Next, remember the overline notation from Section 14.4 of \cite{cla4}. We adopt that notation without any changes, except that $\mathbb{L}$ in it means our present (rather than the old) $\mathbb{L}$. So, for any   formula $E$  including $X$,   
\(\overline{E}\label{ipver}\)
is the result of replacing in $E$ every politeral $L$ by $L\mld\mathbb{L}$. 

The following two lemmas are  proven exactly as the corresponding lemmas in \cite{cla4}: 

\begin{lemma}\label{august12a}
Lemma 14.2 of \cite{cla4} continues to hold. That is: 

For any  formula $E$, including $X$, we have  $E^\dagger=\overline{E}^\dagger$. 
\end{lemma}

\begin{lemma}\label{jan4d}
Lemma 14.3 of \cite{cla4} continues to hold. That is: 

For any    formula $E$, $\arfive\vdash \mathbb{L} \mli \cla \overline{E}$. 
\end{lemma}

In view of Lemma \ref{august12a}, what now remains to do for the completion of our completeness proof is to show that $\arfour\vdash\overline{X}$.

We encode configurations as in Appendix A of \cite{cla4}. 
Also remember the meanings (which we adopt here without any changes) of ``legitimate configuration'', ``yield'' and ``deterministic successor'' from Section 14.5 of \cite{cla4}. As in \cite{cla4}, terminologically we often identify configurations with their codes. For instance, we may say ``configuration $z$'' where what is really meant is ``the configuration encoded by $z$''. 

Based on reasons similar to those relied upon in the proof of Lemma \ref{m29a}, we can write a function $\chi'(z)$ (fix it) polynomial in $z$ (and hence exponential in the size $|z|$ of $z$)\footnote{Asymptotically, $\chi'(z)$ is $O(\mathfrak{s}^{\chi(|z|)})$, where $\mathfrak{s}$ is the number of  symbols of which configurations are composed (see Section A.1 of \cite{cla4}).} such that the following lemma holds: 

\begin{lemma}\label{m31a}
($\pa\vdash:$) Suppose $z$ is  a legitimate configuration, and   $\cal X$ moves in the {\em $i$th} (some $i\geq 0$) deterministic successor of $z$.  Then, as long as $\cal X$ (indeed) runs in space $\chi$,  $i< \chi'(z)$.  
\end{lemma}

Let $E(\vec{s})$ be a   formula all of whose  free variables are among  $\vec{s}$ (but not necessarily vice versa), and $z$ be a variable not among $\vec{s}$. We 
 will write   $E^\circ(z,\vec{s})$ 
to denote an elementary formula whose free variables are $z$ and those of $E(\vec{s})$,  and which is a natural arithmetization of the predicate that,  for any constants $a,\vec{c}$ in the roles of $z,\vec{s}$, holds (that is, $E^\circ(a,\vec{c})$ is true) iff $a$ is a legitimate configuration and its yield is $E(\vec{c})$.
Further,  we 
 will write   $E^{\circ}_{\circ}(z,\vec{s})$ 
to denote an elementary formula whose free variables are $z$ and those of $E(\vec{s})$, and which  is a natural arithmetization of the predicate that,  for any constants $a,\vec{c}$  in the roles of $z,\vec{s}$, holds iff $E^\circ(a,\vec{c})\mlc E^\circ(b,\vec{c})$ is true, where $b$ is the   $\chi'(a)$th deterministic successor of $a$.

 Thus, our present $E^\circ(a,\vec{c})$ means virtually the same as in Section 14.5 of \cite{cla4}, while our present $E^{\circ}_{\circ}(z,\vec{s})$ is a modified version of the $E^{\circ}_{\circ}(z,\vec{s})$ of \cite{cla4}. Namely,  now  $b$ is the  $\chi'(a)$th (rather than $\chi(a)$th as in \cite{cla4}) deterministic successor of $a$.

\begin{lemma}\label{august20b}
Lemma 14.4 of \cite{cla4} continues to hold. That is: 

Assume $E(\vec{s})$ is a non-critical  formula all of whose  free variables are among  $\vec{s}$. 
Then
\[\pa\vdash \cla\bigl(E^{\circ}_{\circ}(z,\vec{s})  \mli \elz{\overline{E(\vec{s})}}\bigr).\] 
\end{lemma}

\begin{proof} Assume the conditions of the lemma. Argue in $\pa$. Consider arbitrary $(\cla$) values of $z$ and $\vec{s}$, which we continue writing as $z$ and $\vec{s}$. Suppose, for a contradiction, that $E^{\circ}_{\circ}(z,\vec{s})$ is true but $\elz{\overline{E(\vec{s})}}$ is false.   The falsity of $\elz{\overline{E(\vec{s})}}$ implies the falsity of $\elz{E(\vec{s})}$. This is so because the only difference between the two formulas is that, wherever the latter has some politeral $L$, the former has a $\mld$-disjunction containing $L$ as a disjunct.  

The truth of $E^{\circ}_{\circ}(z,\vec{s})$ implies that $\cal X$ reaches the configuration (computation step) $z$ and, in the scenario where Environment does not move, $\cal X$ does not move either for at least $\chi'(z)$ steps afterwards. If  $\cal X$ does not move even after $\chi'(z)$ steps, then it has lost the game, because the eventual position hit in the play is $E(\vec{s})$ and the elementarization of the latter is false (as observed in \cite{cla4}, every such game is lost). And if $\cal X$ does make a move sometime after $\chi'(z)$ steps, then, in view of Lemma \ref{m31a},  $\cal X$  does not run in space $\chi$.  Thus, in either case, $\cal X$ does not win $X$ in space $\chi$, that is,  
$\mathbb{L}$ is true.
The rest of this proof is identical to what follows claim (23) in the proof of Lemma 14.4 of \cite{cla4}.  
\vspace{-7pt}  
\end{proof}

\begin{lemma}\label{august20a}
Lemma 14.5 of \cite{cla4} continues to hold. That is: 

Assume $E(\vec{s})$ is a  critical   formula all of whose  free variables are among  $\vec{s}$. Then
\begin{equation}\label{m13a}
\arfive\vdash   \cle E^{\circ}_{\circ} (z,\vec{s})  \mli \cla\overline{E(\vec{s})}.
\end{equation} 
\end{lemma}

\begin{proof} The same as the proof of Lemma 14.5 of \cite{cla4}, taking into account that Lemma 14.3 of \cite{cla4} on which the latter relies continues to hold in our present case according to Lemma \ref{jan4d}.\vspace{-7pt}   
\end{proof}

In the subsequent two lemmas, the notational and terminological conventions of Appendix A of \cite{cla4} are adopted without any changes whatsoever. 

 \begin{lemma}\label{m31b}
Lemmas A.1 through A.6 of \cite{cla4} continue to hold (with ``$\arfive$'' instead of ``$\arfour$'', of course).
\end{lemma}

\begin{proof} The proofs of those lemmas given in \cite{cla4} go through without any changes, as the differences between our present context and the context of \cite{cla4} are fully irrelevant to them.  
\end{proof}

 \begin{lemma}\label{m31c}
Lemma A.7  of \cite{cla4} holds in our present case in  the following stronger form: \[\arfive\vdash \mathbb{C}\mli \mathbb{A}'(z,r)\add \ade x\mathbb{B}(z,x).\]
\end{lemma}

\begin{proof} This proof, as expected, is very close to the proof of Lemma A.7 of \cite{cla4}. Argue in $\arfive$. By $\arfive$-Induction on $r$, we want to show  
\[\mathbb{C}(z)\mli \ade x\bigl(|x|\leq |z|+ |r| \mlc \mathbb{A}(z,x,r)\bigr) \add \ade x \bigl(|x|\leq (|z|+ |r|)\zero \mlc  \mathbb{B}(z,x)\bigr),\] 
from which the target $\mathbb{C}(z)\mli  \mathbb{A}'(z,r) \add \ade x \mathbb{B}(z,x) $ easily follows by LC.

To solve the base   \(\mathbb{C}(z)\mli \ade x\bigl(|x|\leq |z|+ |0| \mlc \mathbb{A}(z,x,0)\bigr) \add \ade x\bigl(|x|\leq (|z|+ |0|)\zero\mlc  \mathbb{B}(z,x)\bigr)
\), we figure out whether the state of $z$ is a move state or not. If yes, we choose the right $\add$-disjunct; if not, we choose the left $\add$-disjunct. In either case, we further choose the value of $z$ for the variable $x$ and win.  

The inductive step is 
\begin{equation}\label{m4c}
\begin{array}{l}
\Bigl(\mathbb{C}(z)\mli \ade x\bigl(|x|\leq |z|+ |r| \mlc \mathbb{A}(z,x,|r|)\bigr) \add \ade x\bigl(|x|\leq (|z|+ |r|)\zero\mlc  \mathbb{B}(z,x)\bigr)\Bigr)\ \mli \\
\Bigl(\mathbb{C}(z)\mli \ade x\bigl(|x|\leq |z|+ |r\successor| \mlc \mathbb{A}(z,x,|r\successor|)\bigr) \add \ade x\bigl(|x|\leq (|z|+ |r\successor|)\zero\mlc  \mathbb{B}(z,x)\bigr)\Bigr).
\end{array}
\end{equation}
To solve (\ref{m4c}), we wait till Environment selects one of  the two $\add$-disjuncts in the  antecedent. 

If the right $\add$-disjunct is selected, we wait further till a constant $c$ for $x$ is selected there. Then  we  select the right $\add$-disjunct in the consequent, and choose the same $c$ for $x$ in it.  

Suppose now the left $\add$-disjunct is selected in the antecedent of (\ref{m4c}). Wait further till a constant $c$ for $x$ is selected there. We may assume that $\mathbb{A}(z,c,r)$  is true, or else we win the game. Using Lemma A.6 of \cite{cla4} (which continues to hold by our Lemma \ref{m31b}),  we find  the deterministic successor $d$ of the configuration $c$. With a little thought, one can see that the size of $d$ cannot exceed the sum of the sizes of $z$ and $r\successor$ more than twice, so that $|d|\leq (|z|+ |r\successor|)\zero$ holds. We figure out whether the state of $d$ is a move state or not. If not, we select the left $\add$-disjunct in the consequent of  (\ref{m4c}), otherwise,  select the right disjunct. In either case, we further choose $d$ for $x$  and win.   
\vspace{-7pt} 
\end{proof}

\begin{lemma}\label{m2a}
Lemma 14.6 of \cite{cla4} continues to hold. That is: 

Assume $E(\vec{s})$ is a formula all of whose free variables are among $\vec{s}$, and $y$ is a variable not occurring in $E(\vec{s})$. Then:      

(a) For every $(\oo,y)$-development $H_{i}(y,\vec{s})$ of $E(\vec{s})$,  $\arfive$ proves $E^{\circ}_{\circ}  (z,\vec{s})  \mli \ade u H_{i}^{\circ}(u,y,\vec{s})\).

(b) Where $H_1(y,\vec{s}),\ldots,H_n(y,\vec{s})$ are all of the $(\pp,y)$-developments of $E(\vec{s})$, $\arfive$ proves
\begin{equation}\label{m2e}   E^{\circ}  (z,\vec{s})  \mli E^{\circ}_{\circ}  (z,\vec{s})\add \mathbb{L}\add \ade u\ade y H_{1}^{\circ}(u,y,\vec{s})\add\ldots\add\ade u\ade y H_{n}^{\circ}(u,y,\vec{s})  . \end{equation}
\end{lemma}

\begin{proof} The proof of clause (a) of Lemma 14.6, given in Section A.2 of \cite{cla4}, only relies on Lemmas A.4 and A.6 of \cite{cla4}, which (by our Lemma \ref{m31b}), continue to hold in our present case. So, clause (a) of our lemma is taken care of. 

For clause (b), assume its conditions.  
In $\arfive$, we can solve  (\ref{m2e}) as follows. Assume  $E^{\circ}  (z,\vec{s})$, which,  of course, implies $\mathbb{C}(z)$. Since $\chi'(z)$ is a polynomial function, we may assume that it is represented as a legitimate term of the language of $\pa$, so that Fact 12.6 of \cite{cla4} applies.  We compute the value  of  $\chi'(z)$, and use that value to specify $r$ in the resource of  
Lemma \ref{m31c}. As a result, we get the resource
$\mathbb{A}'\bigl(z,\chi'(z)\bigr) \add \ade x  \mathbb{B}(z,x)$. This means that we will either know that $\mathbb{A}'\bigl(z,\chi'(z)\bigr)$ is true, or find a constant $a$ for which we will know that $\mathbb{B}(z,a)$ is true. 

If $\mathbb{A}'\bigl(z,\chi'(z)\bigr)$ is true,  then so is $E^{\circ}_{\circ}  (z,\vec{s})$ and, by choosing the latter, we win (\ref{m2e}). 

Now suppose $\mathbb{B}(z,a)$ is true. Then the desired conclusion can be achieved by literally repeating  the corresponding part of the proof given in Section A.3 of \cite{cla4}, as all lemmas relied upon there continue to hold in our present case as well.
\end{proof}

\begin{lemma}\label{m2c}
Lemma 14.7 of \cite{cla4} continues to hold. That is: 

Assume $E(\vec{s})$ is a formula all of whose free variables are among $\vec{s}$. 
Then $\arfive$ proves $E^{\circ}  (z,\vec{s})  \mli \overline{E(\vec{s})}$.  
\end{lemma}

\begin{proof} The proof of Lemma 14.7 of \cite{cla4} goes through here without any changes, taking into account that Lemmas 14.3, 14.4, 14.5 and 14.6 of \cite{cla4} on which the latter relies continue to hold in our present case according to Lemmas \ref{jan4d}, \ref{august20b},  \ref{august20a} and  \ref{m2a}. 
\end{proof}
 
Now we are ready to claim the target result of this section in exactly the same way as at the end of Section 14 of \cite{cla4}. Namely: Let $a$ be the code of the start configuration of $\cal X$, and $\hat{a}$ be a standard variable-free term representing $a$, such as $0$ followed by $a$  ``$\successor$''s.  Of course, $\pa$ and hence $\arfive$ proves $X^\circ(\hat{a})$. By Fact 12.6 of \cite{cla4}, $\arfive$ proves $\ade z(z=\hat{a})$. 
 By Lemma \ref{m2c},   $\arfive$ also proves $\ada z\bigl(X^{\circ}  (z)  \mli \overline{X}\bigr)$. These three can be seen to imply  $\overline{X}$ by LC. Thus, $\arfive\vdash\overline{X}$, as desired. 

\section{$\arsix$,  a theory of elementary recursive computability}\label{ssel}

The language of theory $\arsix$ is the same as that of  $\arfive$, and so are its axioms and the logical rule LC. In addition, just like $\arfive$, $\arsix$ has a single nonlogical rule, which  we call {\bf $\arsix$-Induction}:
\[\frac{\ada \bigl(F(0)\bigr)\hspace{30pt} \ada\bigl( F(x)\mli F(x\successor)\bigr)}{\ada \bigl(F(x)\bigr)},\]
where $F(x)$ is any exponentially bounded formula.

Thus, the only difference between $\arfive$ and $\arsix$ is that, while the induction rule of the former requires the formula $F(x)$ to be polynomially bounded, 
the induction rule of the latter has the weaker requirement that $F(x)$ should be {\em exponentially bounded}.  Below we explain the precise meaning of this term. 

For a variable $x$, by an {\bf exponential sizebound for $x$} we shall mean a standard   formula of the language of $\pa$ saying that $|x|\leq\tau(y_1,\ldots,y_n)$, where $y_1,\ldots,y_n$ are any variables different from $x$, and $\tau(y_1,\ldots,y_n)$ is any   $(0,\successor,+,\times)$-combination of $y_1,\ldots,y_n$.  For instance, $|x|\leq y+ z$ is an exponential sizebound for $x$, which is a formula of $\pa$ saying that    the size of $x$ does not exceed the sum of  $y$ and $z$.  Now, we say that a formula $F$ is {\bf exponentially bounded} iff 
 every subformula $\ada x G(x)$ (resp. $\ade x G(x)$) of $F$ has the form   $\ada x ( S(x)\mli H(x))$ (resp. $\ade x ( S(x)\mlc H(x))$), where  $S(x)$ is an exponential sizebound for $x$  none of whose free variables is bound by $\cla$ or $\cle$ within $F$.

\begin{fact}\label{a1a}
Every $\arfive$-provable (and hence also every $\arfour$-provable) sentence is provable in $\arsix$.
\end{fact}

\begin{proof} We only need to show that $\arsix$ is closed under  $\arfive$-Induction. So, assume $F(x)$ is a polynomially bounded formula, and $\arsix$ proves (the $\ada$-closures of)  both of the following two premises of $\arfive$-Induction:
\begin{eqnarray}
 & F(0);   & \label{am24c}\\
& F(x)\mli F(x\successor). &  \label{am24e}
\end{eqnarray}
Our goal is to show that $\arsix$ proves (the $\ada$-closure of) $F(x)$, the conclusion of $\arfive$-Induction.

For every variable $z$ and every polynomial sizebound $S(z)$ for $z$ that looks like $|z|\leq \tau(|y_1|,\ldots,|y_n|)$, let $S'(z)$ denote the exponential sizebound $|z|\leq \tau(y_1,\ldots,y_n)$ for $z$.  Further, let $F'(x)$ be the result of simultaneously replacing in $F(x)$:
\begin{itemize}
\item every subformula $\ada z \bigl(S(z)\mli G\bigr)$ by $\ada z\Bigl(S'(z)\mli \bigl(S(z)\mli G)\bigr)\Bigr)$;
\item every subformula $\ade z \bigl(S(z)\mlc G\bigr)$ by $\ade z\Bigl(S'(z)\mlc \bigl(S(z)\mlc G)\bigr)\Bigr)$.
\end{itemize}
Note that $F'(x)$ is an exponentially bounded formula. Further, for each of the above sizebounds $S(z)$, $\pa$ obviously proves $\cla z\bigl(S(z)\mli S'(z)\bigr)$. This fact, together with (\ref{am24c}) and (\ref{am24e}), can be easily seen to imply $F'(0)$ and  $ F'(x)\mli F'(x\successor)$ by LC. Thus, $\arsix$ proves both  
$F'(0)$ and  $ F'(x)\mli F'(x\successor)$. Since $F'(x)$ is an exponentially bounded formula, $\arsix$-Induction applies, by which $\arsix$ proves $F'(x)$. The latter, again in conjunction with    $\cla z\bigl(S(z)\mli S'(z)\bigr)$, can be seen to imply $F(x)$ by LC.
\end{proof}

In what follows, we use $\mbox{\em EXP}(x)$ as an (alternative, linear) notation for the function $2^x$.
 
Remember the concept of an explicit polynomial function $\tau$ from Section 10 of \cite{cla4}. An {\bf explicit elementary recursive function} $\tau=\seq{\tau_{f_1},\ldots,\tau_{f_k}}$ is defined in the same way, with the only difference that now, together with $\successor$, $+$, $\times$ and $f_{1},\ldots,f_{i-1}$, each functional (graph)  $\tau_{f_i}$ is allowed to contain $\mbox{\em EXP}$ as an additional function letter with its standard interpretation.  For instance, $\seq{\mbox{\em EXP}(x+ x)_{f_1},\ f_1(f_1(x))_{f_2}}$ (with its two elements written as graphs) is an explicit elementary recursive function. This term represents --- and hence  we identify it with --- the function $2^{(2^{x+x}+2^{x+x})}$. 
When $\tau$ is an explicit elementary recursive function and $\cal M$ is a $\tau$ time (resp. space) machine, we   say that $\tau$ is an {\bf explicit elementary recursive bound} for the time (resp. space) complexity of $\cal M$.

We say that a given HPM $\cal M$ runs in {\bf elementary recursive time} (resp. {\bf space}) iff there is an (explicit) elementary recursive function $\tau$ such that $\cal M$ runs in time (resp. space) $\tau$. And we say that a given problem has an {\bf elementary recursive solution} iff it has a solution that runs in elementary recursive time. The reason why we omitted the word ``time'' here is that, as it is not hard to see (left as an exercise for the reader),  a problem has an elementary recursive time solution if and only if it has an elementary recursive space solution. 

\begin{theorem}\label{ett1}
An arithmetical problem has an elementary recursive solution iff it is provable in $\arsix$. 

Furthermore, there is an efficient procedure that takes an arbitrary extended $\arsix$-proof of an arbitrary sentence $X$ and constructs a   
 solution of $X$ (of $X^\dagger$, that is) together with an explicit elementary recursive bound for its time complexity. 
\end{theorem}

\begin{proof} The soundness (``if'') part of this theorem will be proven in Section \ref{esectsound}, and the completeness (``only if'') part in Section 
\ref{esectcompl}.\vspace{-7pt}
\end{proof}

\section{The soundness of $\arsix$}\label{esectsound}
As in Section \ref{sectsound}, here we will limit ourselves to proving the first, main part (of the soundness part) of Theorem \ref{ett1}; the ``furthermore''  clause of the theorem can be taken care of in the same way as in the similar proof for $\arfour$ given in \cite{cla4}.
  
 Consider any $\arsix$-provable sentence $X$. We proceed by induction on its proof. 

Assume $X$ is an axiom of $\arsix$. As in the cases of $\arfour$ and $\arfive$, if $X$ is one of Peano axioms, then it is a true  elementary sentence and therefore is won by a machine that makes no moves.  
And if $X$ is  $\ada x\ade y(y= x\successor)$, then it is won by a machine that (for the constant $x$ chosen by Environment for the variable $x$)    computes  the value $a$ of $x+ 1$,  makes the move $a$  and retires in a moveless infinite loop.  

Next, suppose $X$ is obtained from premises $X_1,\ldots,X_n$ by LC. By the induction hypothesis, for each $i\in\{1,\ldots,n\}$, we already have a solution (HPM) ${\cal N}_i$ of $X_i$ together with an explicit elementary recursive bound  $\xi_i$ for the time complexity of ${\cal N}_i$.  We can think of each such HPM ${\cal N}_i$ as an $n$-ary GHPM that ignores its inputs. Then, in view of  clause 2 of Theorem 10.1 of \cite{cla4}, we can (efficiently) construct a solution ${\cal M}(\code{{\cal N}_1},\ldots,\code{{\cal N}_n})$ of $X$, together with an explicit elementary recursive bound $\tau(\xi_1,\ldots,\xi_n)$ for its time complexity.

Finally, assume $X$ is (the $\ada$-closure of) $F(x)$, where $F(x)$ is an exponentially bounded formula, and $X$ is obtained by $\arsix$-Induction on $x$. So, the premises are (the $\ada$-closures of) $F(0)$ and $F(x)\mli F(x\successor)$. By the induction hypothesis, there are HPMs ${\cal N}$ and ${\cal K}$ --- with certain explicit elementary recursive bounds $\xi,\zeta$ for their time complexities, respectively --- that solve these two premises, respectively.  Fix them.  We want to construct a solution $\cal M$ of $F(x)$.

As we did in Section 13 of \cite{cla4} or in Section \ref{sectsound} of the present paper, we replace  ${\cal N}$ and  ${\cal K}$   by their ``{reasonable counterparts}'' ${\cal N}'$ and ${\cal K}'$ and corresponding  explicit elementary recursive bounds $\xi',\zeta'$ for their time complexities. Of course, the meaning of  ``reasonable'' is correspondingly redefined now. Namely, now a move's being unreasonable means that the size of the constant chosen by a player for a variable bound by a bounded choice quantifier violates the conditions imposed on it by the exponential (rather than polynomial as before) sizebound for that variable. For simplicity, as in Section \ref{sectsound}, we replace $\xi',\zeta'$ by the common elementary recursive bound $\phi=\xi'+\zeta'$ for the time complexities of both ${\cal N}'$ and ${\cal K}'$.

For further simplicity considerations, 
as  in (near the end of) Section \ref{sectsound}, we  assume that 
 the environment of $\cal M$  never makes illegal moves, for otherwise $\cal M$ easily detects illegal behavior and, being an automatic winner, retires in an moveless infinite loop. We further assume that the environment of $\cal M$, just like ${\cal N}'$ and ${\cal K}'$, plays reasonably. This will not affect the outcome of the game in $\pp$'s ($\cal M$'s) favor, as $\oo$'s unreasonable moves always result in the corresponding subgame's being lost by $\oo$, anyway. From our description of $\cal M$ it will be clear that, as long as Environment plays legally and reasonably, so does $\cal M$, because all it does is copycatting moves by Environment in the real play and moves by ${\cal N}'$ and ${\cal K}'$ in a series of imaginary (simulated) plays. For the same reason, the imaginary adversaries of the simulated ${\cal N}'$ and ${\cal K}'$ will also play legally and reasonably. To summarize, we (safely) assume that  
 all --- real or imaginary --- machines that we consider, as well as their --- real or imaginary --- adversaries, play legally and reasonably.

To describe $\cal M$, as in Section \ref{sectsound}, assume $x,\vec{v}$ are exactly the free variables of $F(x)$, so that   $F(x)$ can be rewritten as $F(x,\vec{v})$. 
At the beginning,   $\cal M$  waits for Environment to choose constants for the free variables of $F(x,\vec{v})$.   Assume $k$ is the  constant chosen for the variable $x$, and $\vec{c}$ are the constants chosen for $\vec{v}$. From now on, we shall write $F'(x)$ for  $F(x,\vec{c})$. Further, as in Section \ref{sectsound}, we shall write ${\cal H}_{0}$ for the ``machine'' that works just like ${\cal N}'$ does in the scenario where the adversary, at the beginning of the play, has chosen the constants $\vec{c}$ for the variables $\vec{v}$. So, ${\cal H}_{0}$ wins the  game $F'(0)$. Similarly, for any $n\geq 1$, we will write 
${\cal H}_{n}$ for the ``machine'' that works just like ${\cal K}'$ does in the scenario where the adversary, at the beginning of the play, has chosen the constants $\vec{c}$ for the variables $\vec{v}$ and the constant $n- 1$ for the variable $x$. So, ${\cal H}_{n}$ wins the  game $F'(n- 1)\mli F'(n)$. Similarly, we will write ${\cal M}_k$ for the ``machine'' that works just like ${\cal M}$ does after the above event of Environment's having chosen $k$ and $\vec{c}$ for $x$ and $\vec{v}$, respectively. So, in order to complete our description of $\cal M$, it will suffice to simply define ${\cal M}_k$ and say that, after Environment has chosen constants for all free variables of $F(x)$, $\cal M$ continues playing as  ${\cal M}_k$.

The work of ${\cal M}_k$ consists in continuously polling its run tape to see if Environment has made any new moves, combined with simulating, in parallel, one play of $ F'(0) $  by ${\cal H}_0$ and --- for each $n\in\{1,\ldots,k\}$ --- 
one play  of $ F'(n- 1)\mli F'(n) $ by ${\cal H}_n$. 
In this mixture of one real  and $k+ 1$ imaginary plays,   
$\cal M$  synchronizes $k+ 1$ pairs of (sub)games, real or imaginary. Namely:

\begin{itemize}
\item It synchronizes  the consequent of the imaginary play of $F'(k-1)\mli F'(k)$ by ${\cal H}_k$ with the real play of $F'(k)$.  
\item For each $n\in\{1,\ldots,k- 1\}$,  it synchronizes the consequent of the imaginary play of  $F'(n-1)\mli F'(n)$ by ${\cal H}_n$ with the antecedent of the 
 imaginary play of $F'(n)\mli F'(n+ 1)$ by ${\cal H}_{n+ 1}$. 
\item It synchronizes  the  imaginary play of $F'(0)$  by ${\cal H}_0$ with the antecedent of the imaginary play of $F'(0)\mli F'(1)$ by ${\cal H}_1$. 
\end{itemize}

This completes our description of ${\cal M}_k$ and hence of $\cal M$.  
Remembering our assumption that (${\cal N},{\cal K}$ and hence) ${\cal N}',{\cal K}'$ win the corresponding games, it is obvious  that ${\cal M}_k$ wins $F'(k)$ and hence $\cal M$ wins $\ada\bigl(F(x)\bigr)$, as desired. It now remains to show that the time complexity of $\cal M$ is also as desired.

For the rest of this proof, pick and fix an arbitrary play (computation branch) of $\cal M$, and an arbitrary clock cycle $\mathfrak{c}$ on which $\cal M$ makes a move $\alpha$ in the real play of $F(x)$. Let $\hbar$ and $\ell$ be the timecost and the background of this move, respectively. Let 
$k$, $F'(x)$, ${\cal H}_0,\ldots,{\cal H}_k$, ${\cal M}_k$  be as in the description of the work of $\cal M$. Note that $\ell$ is not smaller than the size of the greatest of the constants chosen by Environment for the free variables of $F(x)$. Remembering that all players that we consider play legally and reasonably,  one can easily write an explicit elementary recursive function $\eta(\ell)$ (fix it) such that we have:
\begin{equation}\label{m12a}
\mbox{\em The sizes of no moves ever made by ${\cal M}_k$ or the simulated ${\cal H}_n$ ($0\leq n\leq k$) exceed $\eta(\ell)$.}
\end{equation}   
 For instance, if $F(x)$ is 
\(\ade u\bigl(|u|\leq x\times z\mlc\ada v(|v|\leq u+ x\mli G)\bigl)\) where $G$ is elementary, then $\eta(\ell)$ can be taken to be $\mbox{\em EXP}(\ell)\times
\mbox{\em EXP}(\ell)+ \mbox{\em EXP}(\ell)+ 0\successor\successor\successor\successor$.

The polling, simulation and copycat performed by ${\cal M}_k$ do impose some time overhead, but the latter, as in Section 13 of \cite{cla4}, can be safely  ignored.  We will further pretend that the polling and the several simulations happen in a truly parallel fashion, in the sense that ${\cal M}_k$ spends a single clock cycle on tracing a single computation step of  all $k+ 1$ machines simultaneously, as well as on checking out its run tape to see if Environment has made a new move. If so, the rest of our argument is almost literally the same as in Section 13 of \cite{cla4}. Namely:

Let $\beta_1,\ldots,\beta_m$ be the moves by simulated machines that ${\cal M}_k$ detects by time $\mathfrak{c}$, arranged according to the times $t_1\leq \ldots\leq t_m$ of their detections (which, by our simplifying assumptions, coincide with the timestamps of those moves in the corresponding simulated plays).      
Let   $d=\mathfrak{c}- \hbar$. Let $j$ be the smallest integer among $1,\ldots,m$ such that $t_j\geq d$. Since each simulated machine  runs in time $\phi$, in view of (\ref{m12a}) it is clear that $t_j- d$ does not exceed $\phi\bigl(\eta(\ell)\bigr)$. Nor does $t_{i}- t_{i- 1}$ for any $i$ with 
$j< i\leq m$. Hence $t_m- d\leq  (m- j+ 1)\times \phi\bigl(\eta(\ell)\bigr)$. Since $m,j\geq 1$, let us be generous and simply say that $t_m- d\leq  m\times \phi\bigl(\eta(\ell)\bigr)$. But notice that $\beta_m$ is a move made by ${\cal H}_k$ in the consequent of $F'(k- 1)\mli F'(k)$, immediately (by our simplifying assumptions) copied by ${\cal M}_k$ in the real play when it made its move $\alpha$. In other words, $\mathfrak{c}= t_m$. And $\mathfrak{c}- d= \hbar$. So, $\hbar$ does not exceed $m\times  \phi\bigl(\eta(\ell)\bigr)$.   And, by (\ref{m12a}), the size of $\alpha$ does not exceed $m\times  \phi\bigl(\eta(\ell)\bigr)$, either. But observe that $k\leq 2^\ell$ and that $m$ cannot exceed   $k+ 1$ times the depth  $\mathfrak{d}$ of $F(0)$; therefore,    $m\leq  \mathfrak{d}\times (2^{\ell}+ 1)$. Thus, (as long as we pretend that there is no polling/simulation/copycat overhead) neither the timecost nor the size of $\alpha$ exceed $\mathfrak{d}\times (2^{\ell}+ 1)\times \phi\bigl(\eta(\ell)\bigr)$. 

An upper bound for the above function $\mathfrak{d}\times (2^{\ell}+ 1)\times \phi\bigl(\eta(\ell)\bigr)$, even after ``correcting'' the latter so as to precisely  account for the so far suppressed polling/simulation/copycat overhead, can be expressed 
as an explicit elementary recursive function  $\tau$.   This is exactly the sought explicit elementary recursive bound for the time complexity of $\cal M$.

\section{The extensional completeness of $\arsix$}\label{esectcompl}

We treat $\mbox{\em EXP}(x)$ as a pseudoterm and, when writing ``$z=\mbox{\em EXP}(x)$'' within a formula, it is to be understood as an abbreviation of a standard formula of the language of $\pa$ saying that $z$ equals $2^x$. 
 
\begin{fact}\label{a1c}
{\em $\arsix\vdash  \ade z\bigl(z=\mbox{\em EXP}(x)\bigr)$}.
\end{fact}

\begin{proof} Argue in $\arsix$. By $\arsix$-Induction on $x$, we want to prove  $ \ade z\bigl(|z|\leq x\successor \mlc  z=\mbox{\em EXP}(x) \bigr)$, which immediately implies the target $ \ade z\bigl(z=\mbox{\em EXP}(x)\bigr)$ by LC.

The basis  $ \ade z\bigl(|z|\leq 0\successor \mlc  z=\mbox{\em EXP}(0) \bigr)$ is solved by computing the value of $0\successor$ and choosing that value for $z$, which yields the true $ |0\successor|\leq 0\successor \mlc  0\successor =\mbox{\em EXP}(0) $.

To solve the inductive step 
\(\ade z\bigl(|z|\leq x\successor \mlc  z=\mbox{\em EXP}(x)\bigr) \mli \ade z\bigl(|z|\leq x\successor\successor \mlc  z=\mbox{\em EXP}(x\successor)\bigr) ,\)
we wait till Environment chooses a value $a$ for $z$ in the antecedent. Then we compute the value of $a\zero$ and choose that value for $z$ in the consequent, yielding the true 
\(|a|\leq x\successor \mlc  a=\mbox{\em EXP}(x) \mli |a\zero|\leq x\successor\successor \mlc  a\zero=\mbox{\em EXP}(x\successor).\)
\end{proof}

By an {\bf elementary recursive tree-term} we mean a term of the language of $\cltw$ (but not necessarily one of the languages of $\arfour$-$\arsix$) containing no constants other than $0$, and no function symbols other than $\successor$ (unary), $+$ (binary), $\times$ (binary), $\mbox{\em EXP}$ (unary). Every such term represents an elementary recursive function (of the same {\bf arity} as the number of variables in the term) under the standard meaning of its function symbols and $0$. We treat every $n$-ary elementary recursive tree-term $\tau(x_1,\ldots,x_n)$ (can be simply written as $\tau$ instead) as a pseudoterm of the language of $\arsix$ and, when writing ``$z=\tau(x_1,\ldots,x_n)$'', it is to be understood as an abbreviation of a standard formula of   
 $\pa$ saying that $z$ equals the value of $\tau(x_1,\ldots,x_n)$. Such a formula is ``standard'' in the sense that $\pa$ knows the construction of  $\tau$. That is, for instance, if the root   of $\tau(x)$ has the label $+$ and the two tree-terms rooted at the children of the root are $\theta_1(x)$ and $\theta_2(x)$,  then   $\pa\vdash \cla x\bigl(\tau(x)= \theta_1(x)+\theta_2(x)\bigr) $. 

Every explicit elementary recursive function $\tau$ can be translated, in a standard way, into an equivalent (``equivalent'' in the sense of representing the same function) unary elementary recursive tree-term, which we here shall denote by $\tau^\bullet$. This allows us to identify $\tau$ with $\tau^\bullet$ and treat the former, just like the latter, as a pseudoterm. Namely, when writing ``$z=\tau$'' within a formula of the language of $\arsix$, it is to be understood as ``$z=\tau^\bullet$''.  

\begin{fact}\label{a1d}
For any explicit elementary recursive function $\tau$ (not containing $z$),  $\arsix\vdash  \ade z(z=\tau)$.
\end{fact}

\begin{proof} As we know, ``$\ade z(z=\tau)$'' means ``$\ade z(z=\sigma)$'', where $\sigma=\tau^\bullet$.  We prove $\arsix\vdash  \ade z(z=\sigma)$ by (meta)induction on the complexity of the elementary recursive tree-term $\sigma$. The base cases and the cases of $\sigma$ being $\theta\successor$, $\theta_1+\theta_2$ or $\theta_1\times \theta_2$ are handled as in the proof of Lemma 12.6 of \cite{cla4}. The case of $\sigma$ being $\mbox{\em EXP}(\theta)$ is also similar. Namely, by the induction hypothesis, $\arsix$ proves $\ade z(z= \theta)$. And, by Fact \ref{a1c}, $\arsix$ also proves $\ade z\bigl(z=\mbox{\em EXP}(x)\bigr)$. These two easily imply the desired $\ade z\bigl(z= \mbox{\em EXP}(\theta)\bigr)$ by LC.
\end{proof}

The rest of this section is devoted to a proof of the extensional completeness of $\arsix$. Our argument here is almost the same as in Section \ref{sectcompl}, which, in turn, as we remember, mainly consisted in showing that all relevant lemmas employed in the similar completeness proof of $\arfour$ given in \cite{cla4} continued to hold in the present case as well.

We pick an arbitrary elementary-recursively-solvable arithmetical problem $A$ and a sentence  $X$ with $A=X^\dagger$. For the rest of this section, we fix  \(X,\)  and fix \({\cal X}\) as an HPM that solves $A$ (and hence $X^\dagger$) in elementary recursive time. Specifically, we assume that $\cal X$ runs in time 
$\chi$, 
where $\chi$, which we also fix for the rest of this section,  is an explicit elementary recursive function.  We also agree that, throughout the rest of this section, ``{\bf formula}'' exclusively means a subformula of $X$, in which some variables may be renamed. 
Our goal is to construct a sentence $\overline{X}$ for which, just like for $X$, we have $A=\overline{X}^\dagger$ and which, perhaps unlike $X$, is  provable in $\arsix$.

Again, remember the sentence $\mathbb{L}$ from Section 14.3 of \cite{cla4}, saying that $\cal X$ does not win $X$ in time $\chi$. We adopt this meaning of $\mathbb{L}$ without any changes (only now $\chi$ is our present $\chi$ rather than that of \cite{cla4}, of course). The overline notation introduced in Section 
14.4 of \cite{cla4} also retains its old meaning without any changes. And the same holds for the single-circle and double-circle notations $E^{\circ}(z,\vec{s})$ and $E^{\circ}_{\circ}(z,\vec{s})$ introduced in Section 14.5 of \cite{cla4}. Since all relevant concepts here are the same as in \cite{cla4}, we have:

\begin{lemma}\label{a1e}
Lemmas 14.2 through 14.5 of \cite{cla4} continue to hold in our present case (with ``$\arsix$'' instead of ``$\arfour$'').  
\end{lemma}

Namely, according to Lemma 14.2 of \cite{cla4}, we have $X^\dagger= \overline{X}^\dagger$. So,  what now remains to do for the completion of our completeness proof is to show that $\arsix\vdash\overline{X}$.

In the subsequent two lemmas, the notational and terminological conventions of Appendix A of \cite{cla4} are adopted without any changes. 

 \begin{lemma}\label{em31b}
Lemmas A.1 through A.7 of \cite{cla4} continue to hold (with ``$\arsix$'' instead of ``$\arfour$'').
\end{lemma}

\begin{proof} The proofs of those lemmas given in \cite{cla4} go through without any changes, as the differences between our present context and the context of \cite{cla4} are fully irrelevant to them.  
\end{proof}

\begin{lemma}\label{em2a}
Lemma 14.6 of \cite{cla4} continues to hold. That is: 

Assume $E(\vec{s})$ is a formula all of whose free variables are among $\vec{s}$, and $y$ is a variable not occurring in $E(\vec{s})$. Then:      

(a) For every $(\oo,y)$-development $H_{i}(y,\vec{s})$ of $E(\vec{s})$,  $\arsix$ proves $E^{\circ}_{\circ}  (z,\vec{s})  \mli \ade u H_{i}^{\circ}(u,y,\vec{s})\).

(b) Where $H_1(y,\vec{s}),\ldots,H_n(y,\vec{s})$ are all of the $(\pp,y)$-developments of $E(\vec{s})$, $\arsix$ proves
\[E^{\circ}  (z,\vec{s})  \mli E^{\circ}_{\circ}  (z,\vec{s})\add \mathbb{L}\add \ade u\ade y H_{1}^{\circ}(u,y,\vec{s})\add\ldots\add\ade u\ade y H_{n}^{\circ}(u,y,\vec{s})  .\]
\end{lemma}

\begin{proof} The proof of clause (a) of Lemma 14.6, given in Section A.2 of \cite{cla4}, only relies on Lemmas A.4 and A.6 of \cite{cla4}, which (by our Lemma \ref{em31b}), continue to hold in the present case. So, clause (a) of our lemma is taken care of. 

The proof of clause (b) of Lemma 14.6, given in Section A.3 of \cite{cla4}, only relies on Lemmas A.3, A.5, A.6, A.7 and Fact 12.6 of \cite{cla4}. By our Lemma \ref{em31b}, those lemmas   continue to hold in the present case. And Fact 12.6 of \cite{cla4}  we now  replace by Fact \ref{a1d} of the present paper. With this adjustment, the rest of this proof is virtually the same as the proof given in Section A.3 of \cite{cla4}. So, clause (b) of our lemma is also taken care of. 
\end{proof}

\begin{lemma}\label{em2c}
Lemma 14.7 of \cite{cla4} continues to hold. That is: 

Assume $E(\vec{s})$ is a formula all of whose free variables are among $\vec{s}$. 
Then $\arsix$ proves $E^{\circ}  (z,\vec{s})  \mli \overline{E(\vec{s})}$.  
\end{lemma}

\begin{proof} The proof of Lemma 14.7 of \cite{cla4} goes through here without any changes, taking into account that Lemmas 14.3, 14.4, 14.5 and 14.6 of \cite{cla4} on which the latter relies continue to hold in our present case according to Lemmas \ref{a1e} and  \ref{em2a}. 
\end{proof}
 
Now we can claim the target result of this section in exactly the same way as in the last paragraph of Section 14 of \cite{cla4} or the last paragraph of 
Section \ref{sectcompl} of the present paper.

\section{$\arseven$,  a theory of primitive recursive computability}\label{psel}

The language of  $\arseven$ is the same as those of  $\arfive,\arsix$, and so are its axioms and the logical rule LC. In addition, just like $\arfive$ and $\arsix$, theory $\arseven$ has a single nonlogical rule, which  we call {\bf $\arseven$-Induction}:
\[\frac{\ada \bigl(F(0)\bigr)\hspace{30pt} \ada\bigl( F(x)\mli F(x\successor)\bigr)}{\ada \bigl(F(x)\bigr)},\]
where $F(x)$ is any formula.

Thus, the only difference between $\arseven$ and $\arfive$ or $\arsix$ is that, while the induction rules of the latter require the formula $F(x)$ to be polynomially or exponentially bounded, 
the induction rule of the former imposes  no restrictions on $F(x)$ at all.

\begin{fact}\label{pa1a}
Every $\arsix$-provable (and hence also every $\arfour$-provable and every $\arfive$-provable) sentence is provable in $\arseven$.
\end{fact}

\begin{proof} This is straightforward, as $\arsix$-Induction is a special case of $\arseven$-Induction.
\end{proof}

Let $f$ be a function letter of the language of $\cltw$ of indicated (by the number of explicitly shown arguments) arity. 

An {\bf absolute primitive recursive definition} of $f$ is a $\cltw$-formula of one of the following forms:

\begin{quote}
\begin{description}
\item[(I)] $\cla x \bigl(f(x)= x'\bigr)$.
\item[(II)] $\cla x_1\ldots\cla x_n \bigl(f(x_1,\ldots,x_n)=0\bigr)$.
\item[(III)] $\cla x_1\ldots\cla x_n \bigl(f(x_1,\ldots,x_n)= x_i\bigr)$ (some $i\in\{1,\ldots,n\}$).
\end{description}
\end{quote}

And a {\bf relative primitive recursive definition} of $f$ is a $\cltw$-formula of one of the following forms:

\begin{quote}
\begin{description}
\item[(IV)] $\cla x_1\ldots\cla x_n \Bigl(f(x_1,\ldots,x_n)= g\bigl(h_1(x_1,\ldots,x_n),\ldots,h_m(x_1,\ldots,x_n)\bigr)\Bigr)$.
\item[(V)] $\begin{array}{l}
\cla x_2\ldots\cla x_n \bigl(f(0,x_2,\ldots,x_n)= g(x_2,\ldots,x_n)\bigr)\ \mlc
\\
\cla x_1 \cla x_2\ldots\cla x_n \Bigl(f(x'_1,x_2,\ldots,x_n)= h\bigl(x_1,f(x_1,x_2,\ldots,x_n),x_2,\ldots,x_n\bigr)\Bigr).
\end{array}$
\end{description}
\end{quote}

We say that (IV) defines $f$ {\bf in terms of} $g,h_1,\ldots,h_m$. Similarly, we say that (V) defines $f$ in terms of $g$ and $h$. 

A {\bf primitive recursive construction} of $f$ is a sequence $E_1,\ldots,E_k$ of $\cltw$-formulas, where each $E_i$ is a primitive recursive definition of some $g_i$, all such $g_i$ are distinct, $g_k= f$ and, for each $i$, $E_i$ is either an absolute primitive recursive definition of $g_i$, or a relative primitive recursive definition of $g_i$ in terms of some $g_j$s with $j< i$.

Terminologically, we will usually identify a primitive recursive construction of a function $f$ with the function $f$ itself. Further, to keep our terminology uniform,  we will be using the words ``{\bf explicit primitive recursive function}'' as a synonym of ``primitive recursive construction of a unary function''.    
When $\tau$ is an explicit primitive recursive function and $\cal M$ is a $\tau$ time (resp. space) machine, we  say that $\tau$ is an {\bf explicit primitive recursive bound} for the time (resp. space) complexity of $\cal M$.

We say that a given HPM $\cal M$ runs in {\bf primitive recursive time} (resp. {\bf space}) iff there is an explicit primitive recursive function $\tau$ such that $\cal M$ runs in time (resp. space) $\tau$. And we say that a given problem has a {\bf primitive recursive solution} iff it has a solution that runs in primitive recursive time. The reason why we omitted the word ``time'' here is that, as in the case of elementary recursiveness,  it is not hard to see  that a problem has a primitive recursive time solution if and only if it has a primitive  recursive space solution. 

\begin{theorem}\label{ptt1}
An arithmetical problem has a primitive recursive solution iff it is provable in $\arseven$. 

Furthermore, there is an efficient procedure that takes an arbitrary extended $\arseven$-proof of an arbitrary sentence $X$ and constructs a   
 solution of $X$ (of $X^\dagger$, that is) together with an explicit primitive recursive bound for its time complexity. 
\end{theorem}

\begin{proof} The soundness (``if'') part of this theorem will be proven in Section \ref{psectsound}, and the completeness (``only if'') part in Section 
\ref{psectcompl}.\vspace{-7pt}
\end{proof}

\section{The soundness of $\arseven$}\label{psectsound}
As in Sections \ref{sectsound} and \ref{esectsound}, we will limit ourselves to proving the pre-``furthermore'' part (of the soundness part) of Theorem \ref{ptt1}. 
 Consider any $\arseven$-provable sentence $X$. We proceed by induction on its proof. 

The case of  $X$ being an axiom is handled in the same way as in the soundness proofs for the previous systems. So is the case of $X$ being obtained by LC. 

For the rest of this section, suppose $X$ is (the $\ada$-closure of) $F(x)$, and $X$ is obtained by $\arseven$-Induction on $x$. So, the premises are (the $\ada$-closures of) $F(0)$ and $F(x)\mli F(x\successor)$. By the induction hypothesis, there are HPMs ${\cal N}$ and ${\cal K}$ --- with certain explicit primitive recursive bounds $\xi,\zeta$ for their time complexities, respectively --- that solve these two premises, respectively.  We replace  $\xi$ and $\zeta$ by one common bound $\phi= \xi+ \zeta$ for the time complexities of both ${\cal N}$ and ${\cal K}$.

As in Section \ref{esectsound},   we will assume that the adversary of the purported solution $\cal M$ of $F(x)$ that we are going to construct  never makes illegal moves.  From our description of $\cal M$ it will be clear that, as long as Environment plays legally,   so does $\cal M$.

To describe  $\cal M$, as before, assume $x,\vec{v}$ are exactly the free variables of $F(x)$, so that   $F(x)$ can be rewritten as $F(x,\vec{v})$. 
At the beginning,   $\cal M$  waits for Environment to choose constants for the free variables of $F(x,\vec{v})$.   Assume $k$, with $k\geq 1$ (the case of $k= 0$ is straightforward),  is the  constant chosen for the variable $x$, and $\vec{c}$ are the constants chosen for $\vec{v}$. From now on, we shall write $F'(x)$ for $F(x,\vec{c})$. Further, as in Sections \ref{sectsound} and \ref{esectsound}, we shall write ${\cal H}_{0}$ for the ``machine'' that works just like ${\cal N}$ does in the scenario where the adversary, at the beginning of the play, has chosen the constants $\vec{c}$ for the variables $\vec{v}$. So, ${\cal H}_{0}$ wins  $F'(0)$. Similarly, for any $n\geq 1$, we will write 
${\cal H}_{n}$ for the ``machine'' that works just like ${\cal K}$ does in the scenario where the adversary, at the beginning of the play, has chosen the constants $\vec{c}$ for the variables $\vec{v}$ and the constant $n- 1$ for the variable $x$. So, ${\cal H}_{n}$ wins  $F'(n- 1)\mli F'(n)$. Similarly, we will write ${\cal M}_k$ for the ``machine'' that works just like ${\cal M}$ does after the above event of Environment's having chosen $k$ and $\vec{c}$ for $x$ and $\vec{v}$, respectively. So, in order to complete our description of $\cal M$, it will suffice to simply define ${\cal M}_k$ and say that, after Environment has chosen constants for all free variables of $F(x)$, $\cal M$ continues playing as  ${\cal M}_k$.

The work of ${\cal M}_k$ consists in continuously polling its run tape to see if Environment has made any new moves, combined with simulating, in parallel, one play of $ F'(0) $  by ${\cal H}_0$ and --- for each $n\in\{1,\ldots,k\}$ --- 
one play  of $ F'(n- 1)\mli F'(n) $ by ${\cal H}_n$. 
In this mixture of one real  and $k+ 1$ imaginary plays,   
$\cal M$  synchronizes $k+ 1$ pairs of (sub)games, real or imaginary. Namely:

\begin{itemize}
\item It synchronizes  the consequent of the imaginary play of $F'(k- 1)\mli F'(k)$ by ${\cal H}_k$ with the real play of $F'(k)$.  
\item For each $n\in\{1,\ldots,k- 1\}$,  it synchronizes the consequent of the imaginary play of  $F'(n- 1)\mli F'(n)$ by ${\cal H}_n$ with the antecedent of the 
 imaginary play of $F'(n)\mli F'(n+ 1)$ by ${\cal H}_{n+ 1}$. 
\item It synchronizes  the  imaginary play of $F'(0)$  by ${\cal H}_0$ with the antecedent of the imaginary play of $F'(0)\mli F'(1)$ by ${\cal H}_1$. 
\end{itemize}

This completes our description of ${\cal M}_k$ and hence of $\cal M$.   
Remembering our assumption that ${\cal N},{\cal K}$ win the corresponding games, it is obvious  that ${\cal M}_k$ wins $F'(k)$ and hence $\cal M$ wins $\ada\bigl(F(x)\bigr)$, as desired. It now remains to show that the time complexity of $\cal M$ is also as desired.

For the rest of this proof, pick and fix an arbitrary play   of $\cal M$, and an arbitrary clock cycle $\mathfrak{c}$ on which $\cal M$ makes a move $\alpha$ in the real play of $F(x)$. Let $\hbar$ and $\ell$ be the timecost and the background of this move, respectively. Let 
$k$, $F'(x)$, ${\cal H}_0,\ldots,{\cal K}_k$, ${\cal M}_k$  be as in the description of the work of $\cal M$. 

As done before in similar proofs, we ignore the polling, simulation and copycat overhead, and also  pretend that the polling and the several simulations happen in a truly parallel fashion, in the sense that $\cal M$ spends a single clock cycle on tracing a single computation step of  all $k+ 1$ machines simultaneously, as well as on checking out its run tape to see if Environment has made a new move.

Let $\beta_1,\ldots,\beta_m$ be the moves by simulated machines  that ${\cal M}_k$ detects by time $\mathfrak{c}$, arranged according to the times $t_1\leq \ldots\leq t_m$ of their detections (which, by our simplifying assumptions, coincide with the timestamps of those moves in the corresponding simulated plays).       
Let   $d=\mathfrak{c}- \hbar$. Let $j$ be the smallest integer among $1,\ldots,m$ such that $t_j\geq d$. Since each simulated machine runs in time $\phi$, it is clear that neither the size of $\beta_j$ nor $t_j- d$ exceed $\phi(\ell)$. For similar reasons, with $\phi(\ell)$ now acting in the role of $\ell$, neither the size of $\beta_{j+ 1}$ nor $t_{j+ 1}- t_j$ exceed $\phi(\phi(\ell))$. Therefore, neither the size of $\beta_{j+ 1}$ nor $t_{j+ 1}- d$ exceed $2\phi(\phi(\ell))$. Similarly, neither the size of $\beta_{j+ 2}$ nor $t_{j+ 2}- d$ exceed $3\phi(\phi(\phi(\ell)))$. And so on. Thus, neither the size of $\beta_{m}$ nor $t_{m}- d$ exceed $(m- j+ 1)\times \phi^{m- j+ 1}(\ell)$ and hence (as $m,j\geq 1$) 
$m\times \phi^{m}(\ell)$, where $\phi^{m}$ means the $m$-fold composition of $\phi$ with itself. Also note that $m$ cannot exceed $2^{\ell}\times \mathfrak{d}$, where $\mathfrak{d}$ is the depth of $F(x)$. We conclude that neither the size of $\beta_m$ nor $\hbar$ exceed $2^{\ell}\times \mathfrak{d}\times \phi^{2^{\ell}\times \mathfrak{d}}(\ell)$.  
But notice that $\beta_m$ is a move made by ${\cal H}_k$ in the consequent of $F'(k-1)\mli F(k)$, immediately (by our simplifying assumptions) copied by ${\cal M}_k$ in the real play when it made its move $\alpha$. In other words, $\mathfrak{c}= t_m$.  Thus, (as long as we pretend that there is no polling/simulation/copycat overhead) neither the timecost nor the size of $\alpha$ exceed  $2^{\ell}\times \mathfrak{d}\times \phi^{2^{\ell}\times \mathfrak{d}}(\ell)$.  

Obviously an upper bound for the above function  $2^{\ell}\times \mathfrak{d}\times \phi^{2^{\ell}\times \mathfrak{d}}(\ell)$, even after ``correcting'' the latter so as to precisely  account for the so far suppressed polling/simulation/copycat overhead, can be expressed 
as an explicit primitive recursive function  $\tau(\ell)$.   This is exactly the sought explicit primitive recursive bound for the time complexity of $\cal M$.

\section{The extensional completeness of $\arseven$}\label{psectcompl}

We treat each  $n$-ary primitive recursive construction  $\tau= \tau(x_1,\ldots,x_n)$ as a pseudoterm and,  when we write ``$z= \tau(x_1,\ldots,x_n)$'' (or just ``$z= \tau$'') within a formula, it is to be understood as an abbreviation of a standard formula of $\pa$ saying that $z$ equals the value of $\tau(x_1,\ldots,x_n)$. Such a formula is ``standard'' in the sense that $\pa$ knows the definition of $\tau$. That is, for instance, if $\tau(x)$ is defined (in its primitive recursive construction) by $\cla x\Bigl(\tau(x)= \theta\bigl(\phi(x)\bigr)\Bigr)$, then   $\pa\vdash \cla x\Bigl(\tau(x)= \theta\bigl(\phi(x)\bigr)\Bigr)$. 

\begin{fact}\label{a2a}
For any explicit primitive recursive function $\tau$ (not containing $z$),  $\arseven\vdash  \ade z(z=\tau)$.
\end{fact}

\begin{proof} We generalize the above statement by allowing $\tau$ to be a primitive recursive construction of any (not necessarily unary) function, and prove such a generalized statement  by metainduction on the complexity of (the construction of) $\tau$. This requires considering the five cases I-V from Section \ref{psel}, depending on which of them applies last in the construction of $\tau$.

{\em Case I}: $\tau$ is a function defined by $\cla x\bigl(\tau(x)= x\successor\bigr)$. This sentence is thus provable in $\pa$. By Axiom 8, $\arseven$ also proves $\ada x\ade y(y= x\successor)$. The desired $\ade z\bigl(z= \tau(x)\bigr)$ is an easy logical consequence of these two.  

{\em Case II}: $\tau$ is a function defined by $\cla x_1\ldots\cla x_n\bigl(\tau(x_1,\ldots,x_n)= 0\bigr)$. This sentence is thus provable in $\pa$. The target     $\ade z\bigl(z= \tau(x_1,\ldots,x_n)\bigr)$ is an immediate logical consequence of it.

{\em Case III}: $\tau$ is a function defined by $\cla x_1\ldots\cla x_n\bigl(\tau(x_1,\ldots,x_n)= x_i\bigr)$. Similar to the preceding case. 

{\em Case IV}:  $\tau$ is a function defined by $\cla x_1\ldots\cla x_n \Bigl(\tau(x_1,\ldots,x_n)= \phi\bigl(\psi_1(x_1,\ldots,x_n),\ldots,\psi_m(x_1,\ldots,x_n)\bigr)\Bigr)$. This sentence is thus provable in $\pa$. By the induction hypothesis, $\arseven$ also proves $\ade z\bigl(z= \phi(x_1,\ldots,x_m)\bigr)$ and --- for each $i\in\{1,\ldots,m\}$ --- $\ade z\bigl(z= \psi_i(x_1,\ldots,x_n)\bigr)$. These provabilities can be seen to imply the provability of the target  $\ade z\bigl(z= \tau(x_1,\ldots,x_n)\bigr)$ by LC.

{\em Case V}:  $\tau$ is a function defined by
 \begin{equation}\label{a2b}\begin{array}{l}
\cla x_2\ldots\cla x_n \bigl(\tau(0,x_2,\ldots,x_n)= \theta(x_2,\ldots,x_n)\bigr)\ \mlc
\\
\cla x_1 \cla x_2\ldots\cla x_n \Bigl(\tau({x_1}\successor,x_2,\ldots,x_n)= \phi\bigl(x_1,\tau(x_1,x_2,\ldots,x_n),x_2,\ldots,x_n\bigr)\Bigr),
\end{array}\end{equation}
so that the above sentence is provable in $\pa$. By the induction hypothesis, $\arseven$ also proves both of the following:
\begin{eqnarray}
& \ade z\bigl(z= \theta(x_2,\ldots,x_n)\bigr); & \label{a2c}\\
& \ade z\bigl(z= \phi(x_0,\ldots,x_n)\bigr). & \label{a2d}
\end{eqnarray}

By LC, (\ref{a2b}), (\ref{a2c}) and (\ref{a2d}) can be seen to imply both of the following:

\begin{eqnarray}
& \ade z\bigl(z= \tau(0,x_2,\ldots,x_n)\bigr); & \label{a2e}\\
& \ade z\bigl(z= \tau(x_1,x_2\ldots,x_n)\bigr)\mli \ade z\bigl(z= \tau(x_1\successor,x_2\ldots,x_n)\bigr). & \label{a2f}
\end{eqnarray}

Now, the target $\ade z\bigl(z= \tau(x_1,x_2\ldots,x_n)\bigr)$ follows from (\ref{a2e}) and (\ref{a2f}) by $\arseven$-Induction on $x_1$. \end{proof}

The rest of our completeness proof for $\arseven$ is literally the same as the completeness proof for $\arsix$ found in Section \ref{esectcompl}, with the only difference that now $\chi$ is an explicit primitive recursive (rather than elementary recursive) function; also, where Section \ref{esectcompl} relied on Fact \ref{a1d}, now we  rely on Fact \ref{a2a} instead.

\section{On the intensional strength of $\arfive$, $\arsix$ and $\arseven$}\label{sculprit}
The following theorem is proven in literally the same way as Theorem 16.1 of \cite{cla4}:

\begin{theorem}\label{feb15}
Let $X$ and $\mathbb{L}$ be as in Section \ref{sectcompl} (resp. Section \ref{esectcompl}, resp. Section \ref{psectcompl}).  Then $\arfive$ (resp. $\arsix$, resp. $\arseven$) proves $\gneg \mathbb{L}\mli X$.
\end{theorem}

So, whatever was said in Section 16 of \cite{cla4} about the import of this theorem, extends to our present systems $\arfive$, $\arsix$ and $\arseven$ as well. This includes 
Theorem 16.2 of \cite{cla4}. To re-state that theorem, we extend the earlier concept of constructive  provability to our present complexity classes. Namely,  we say that $\pa$ {\bf constructively proves} the polynomial space (resp. elementary recursive, resp. primitive recursive) computability of a sentence $X$ iff, for some particular HPM $\cal X$ and some particular explicit polynomial (resp. elementary recursive, resp. primitive recursive) function $\chi$, $\pa$ proves that $\cal X$ is   a $\chi$-space (resp. $\chi$-time or $\chi$-space, resp. $\chi$-time or $\chi$-space)  solution of $X$.  

Then the following theorem is proven in literally the same way as Theorem 16.2 of \cite{cla4}:

\begin{theorem}\label{jan30}
Let $X$ be any sentence  such that $\pa$ constructively proves  the polynomial space (resp. elementary recursive, resp. primitive recursive) computability of $X$. Then $\arfive$ (resp. $\arsix$, resp. $\arseven$) proves $X$. 
\end{theorem}


\begin{thebibliography}{99}

\bibitem{Buss} S. Buss. {\bf Bounded Arithmetic} (revised version of Ph. D. thesis). Bibliopolis, 1986. 

\bibitem{bbb3}
P. Clote and G. Takeuti.  {\em Bounded arithmetic for NC, ALogTIME, L and NL}. {\bf Annals of Pure and Applied Logic} 56 (1992), pp. 73-117.

\bibitem{Jap03} G. Japaridze. {\em Introduction to computability logic}. {\bf Annals of Pure and Applied Logic} 123 (2003), pp. 1-99.

\bibitem{Japjsl} G. Japaridze. {\em The logic of interactive Turing reduction}. {\bf Journal of Symbolic Logic} 72 (2007), pp. 243-276. 

\bibitem{int1} G. Japaridze. {\em Intuitionistic computability logic}. {\bf Acta Cybernetica} 18 (2007),  pp. 77-113.  

\bibitem{Propint} G. Japaridze. {\em The intuitionistic fragment of computability logic at the propositional level}. {\bf Annals of Pure and Applied Logic} 147 (2007),  pp. 187-227. 

\bibitem{Japfin} G. Japaridze. {\em In the beginning was game semantics}.  {\bf Games: Unifying Logic, Language, and Philosophy}. O. Majer,
A.-V. Pietarinen and T. Tulenheimo, eds. Springer   2009,  pp. 249-350.

\bibitem{Japtowards} G. Japaridze. {\em Towards applied theories based on computability logic}. {\bf Journal of Symbolic Logic} 75 (2010), pp. 565-601.  

\bibitem{cla4} G. Japaridze. {\em Introduction to clarithmetic I}.  {\bf Information and Computation} 209 (2011), pp. 1312-1354. 

\bibitem{Japlbcs} G. Japaridze. {\em A logical basis for constructive systems}. {\bf Journal of Logic and Computation} 22 (2012), pp. 605-642. 



\end{thebibliography}
\end{document}